%% file: main.tex
\documentclass[twoside]{article}
\usepackage{aistats2027}
\usepackage[T1]{fontenc}
\usepackage[utf8]{inputenc}
\usepackage{amsmath,amsthm,mathtools}
\usepackage{aliascnt}
\usepackage{amssymb}
\usepackage{microtype}
\usepackage{booktabs,array,tabularx,longtable}
\usepackage{algorithm}
\usepackage[noend]{algpseudocode}
\usepackage{balance}
\algrenewcommand{\algorithmicrequire}{\textbf{Input:}}
\algrenewcommand{\algorithmicensure}{\textbf{Output:}}

\usepackage{enumitem}
\usepackage{xcolor}
\usepackage{hyperref}
\usepackage[nameinlink,noabbrev]{cleveref}
\usepackage{placeins}
\usepackage[numbers,sort&compress]{natbib}

\makeatletter
\providecommand{\theHALG@line}{}
\renewcommand{\theHALG@line}{\thealgorithm.\arabic{ALG@line}}
\makeatother
\usepackage{needspace}
\hypersetup{colorlinks=true,linkcolor=blue!45!black,citecolor=blue!45!black,
 urlcolor=blue!45!black,pdftitle={BetweenCut: Private Heavy-Node Classification with Doubly Logarithmic Error in Tree Height},
 pdfauthor={Ergute Bao and Tudor Cebere},
 pdfsubject={Explicit-constant privacy and deterministic utility analysis}}
\allowdisplaybreaks[2]
\newtheorem{theorem}{Theorem}[section]
\newaliascnt{lemma}{theorem}
\newtheorem{lemma}[lemma]{Lemma}
\aliascntresetthe{lemma}
\newaliascnt{proposition}{theorem}
\newtheorem{proposition}[proposition]{Proposition}
\aliascntresetthe{proposition}
\newaliascnt{corollary}{theorem}

\aliascntresetthe{corollary}
\theoremstyle{definition}
\newaliascnt{definition}{theorem}
\newtheorem{definition}[definition]{Definition}
\aliascntresetthe{definition}
\theoremstyle{remark}
\newaliascnt{remark}{theorem}

\aliascntresetthe{remark}
\newcommand{\N}{\mathbb{N}}

\newcommand{\Prb}{\Pr}

\newcommand{\Lo}{\mathsf{L}}
\newcommand{\Bt}{\mathsf{B}}
\newcommand{\Hi}{\mathsf{H}}
\newcommand{\Unk}{\mathsf{?}}
\newcommand{\Lap}{\operatorname{Lap}}
\newcommand{\Anc}{\operatorname{Anc}}
\newcommand{\Desc}{\operatorname{Desc}}

\newcommand{\parent}{\operatorname{parent}}

\newcommand{\ind}{\mathbf{1}}
\newcommand{\Hhat}{\widehat H}

\newcommand{\eps}{\varepsilon}
\newcommand{\GBT}{\mathsf{GuardedBT}}
\newcommand{\BT}{\mathsf{BT}}
\newcommand{\qstar}{q_\lambda}

\crefname{algorithm}{Algorithm}{Algorithms}
\Crefname{algorithm}{Algorithm}{Algorithms}
\crefname{lemma}{lemma}{lemmas}
\Crefname{lemma}{Lemma}{Lemmas}
\crefname{proposition}{proposition}{propositions}
\Crefname{proposition}{Proposition}{Propositions}
\crefname{corollary}{corollary}{corollaries}
\Crefname{corollary}{Corollary}{Corollaries}
\crefname{definition}{definition}{definitions}
\Crefname{definition}{Definition}{Definitions}
\crefname{remark}{remark}{remarks}
\Crefname{remark}{Remark}{Remarks}
\begin{document}
\twocolumn[
\aistatstitle{BetweenCut: Private Heavy-Node Classification with\\ Doubly Logarithmic Error in Tree Height}
% Keep the manuscript globally in review mode so that the AISTATS review
% running heads and the "Do not distribute" notice remain in place.  Switch
% to the identity-rendering branch only for the title-page author block.
\begingroup
\renewcommand{\statePaper}{1}
\runningauthor{Manuscript under review}
\aistatsauthor{%
\begin{tabular}{c@{\hspace{0.8in}}c}
Ergute Bao & Graham Cormode\\
{\normalfont Inria} & {\normalfont University of Oxford}\\[0.75em]
Xiaokui Xiao & Ting Yu\\
{\normalfont National University of Singapore} &
{\normalfont Mohamed bin Zayed University of Artificial Intelligence}
\end{tabular}}

\vskip 0.18in % space between author block and abstract

\endgroup
]
\allowdisplaybreaks
\vspace{5em}

\input{0-abs}
\input{1-intro}
\input{2-formulation}

\input{3-related-works}

\input{4-algo}

\input{6-analysis-privacy-trade-off}
\input{7-comparison-short}
\input{8-conclusion}

\input{AI}

\FloatBarrier

\phantomsection
\addcontentsline{toc}{section}{References}

\bibliographystyle{plainnat}
\bibliography{references}

\balance

\clearpage
\appendix
\onecolumn

\input{App-other-related}

\input{App-algo-details}

\input{App-proofs}

\input{App-sparse}

\input{App-extension-tree-counting}

\end{document}

%% file: 0-abs.tex
\begin{abstract}
Finding heavy nodes in a tree---those whose counts exceed a given threshold---is a building block for analysis and learning over structured data. 
Achieving record-level differential privacy (DP) without sacrificing accuracy is challenging because each record contributes to counts along an entire root-to-leaf path, allowing privacy costs to accumulate across levels. 
Existing methods account for the multiple threshold comparisons for each record incur additive error margins of $\Omega_{\varepsilon,\delta}(\log h)$ or $\Omega_{\varepsilon,\delta}(\sqrt{\log h})$ for tree height $h$. 
We introduce \textsc{BetweenCut}, an $(\varepsilon,\delta)$-DP algorithm with an additive error margin of $O_{\varepsilon,\delta}(\log\log h)$, improving the existing bounds for deep trees. 
This error holds simultaneously for all nodes and is independent of the input database size.

\end{abstract}

%% file: 1-intro.tex
\section{Introduction}\label{sec:intro}

Many data spaces are naturally represented (implicitly or explicitly) as a hierarchy~\cite{zhu2020federated,kim2021ngram,wu2022heavyhitters,practical-ldp-hh,apple2017learning}. 
Searching for important nodes---those whose ``weight'' exceeds a given threshold $\tau$---is a building block for analysis and learning over structured data, such as random forests~\cite{lumberjack2026,ferry-random-forest-attack}, frequent word discovery~\cite{zhu2020federated}, and spatial partitioning pruning~\cite{zhang2016privtree}. 
We study this task on a public hierarchical tree of height $h$, where each private record belongs to one leaf and contributes to every count on its root-to-leaf path. 
The goal is to classify all nodes simultaneously while satisfying record-level $(\varepsilon,\delta)$-differential privacy ($(\varepsilon,\delta)$-DP)~\cite{dwork:roth:14}, with an additive error margin around $\tau$ that remains small even when $h$ is large. \looseness=-1

A natural baseline releases a DP node histogram independently at each level and thresholds the released counts against $\tau$. Within a level, the histogram bins contain disjoint sets of records. Across levels, however, the same record participates in $h$ levels. The basic Laplace mechanism therefore gives noise scale $h/\varepsilon$, while the Gaussian mechanism gives scale proportional to $\sqrt{h}/\varepsilon$ at fixed $\delta$, both leading to errors in $O(\operatorname{poly}(h))$. \looseness=-1
% These polynomial height factors can be further amplified by requiring simultaneous accuracy guarantees over all nodes, leading to errors in $poly(h)$.

One approach is to \emph{reduce the number of threshold tests}. Lumberjack~\cite{lumberjack2026} searches middle levels first and propagates `Heavy' labels upward and `Light' labels downward. This balanced search reduces each record's participation from $h$ to
$O(\log h)$ tests and leads to an improved error margin of $O_{\varepsilon,\delta}(\sqrt{\log h\,\log\log h})$.

A different approach is to \emph{avoid a separate privacy charge for every test}. 
The sparse vector technique (SVT) controls privacy through a bounded number of positive answers rather than the total number of query participation~\cite{svt-origin,dwork:roth:14}.
Ghazi et al.~\cite{ghazi2023} use this principle in a bottom-up tree classifier, and Biswas et al.~\cite{biswas2026} develop per-record threshold-monitoring arguments for hierarchical heavy hitters.
These approaches can traverse many levels without a privacy budget proportional to their number, but their stated errors retain $\log h$ terms due to a union bound over all levels. 

This suggests a question: \emph{can a tree search make a small number of queries per record while avoiding a separate privacy charge for each query?} 
These two requirements are distinct. The label propagation idea in Lumberjack~\cite{lumberjack2026} reduces the number of queries, but each query still triggers a separate privacy charge: a `High' answer leaves descendants to be classified, whereas a `Low'
answer can leave ancestors unresolved; either answer may still recur along the same root-to-leaf path for a record. On the other hand, applying SVT-style threshold tests to hierarchal nodes still requires composing the privacy costs for all positive answers: Ghazi et al.~\cite{ghazi2023} imposes a global $O(M/\tau)$ positive-answer budget, where $M\ge n$ is a public input database-size bound, introducing a different source of error. \looseness=-1

We provide a positive answer. Our algorithm fuses these two previously separate lines of work: histogram-based node counting~\cite{lumberjack2026} and threshold-testing~\cite{ghazi2023,biswas2026} with SVT accounting. We make a near-threshold answer terminate every affected root-to-leaf path through the queried node. This threshold test distinguishes `Low', `High', and `Middle' outcomes using fresh Laplace noise and two fixed public thresholds around $\tau$. `Low' and `High' outcomes propagate labels in one direction. A `Middle' outcome does both: it labels every proper ancestor Heavy, and the queried node and all its descendants Light. In other words, the `Middle' outcome terminates the search---there can be at most one `Middle' outcome in any root-to-leaf path. 

In our analysis, we argue that the `Low' and `High' outcomes need not trigger privacy charges. This requires the introduction of a novel ``one-target, finite horizon'' variant of the Target-charging technique (TCT)~\cite{cohenlyu2023}, which itself generalizes the Sparse Vector Technique (SVT)~\cite{svt-origin} and BetweenThresholds~\cite{bnsv2015}.
This analysis permits reusing a noiseless threshold across multiple queries on disjoint paths, simplifying the algorithm and avoiding re-sampling and composition overheads. 
The co-creation of the algorithm and its analysis was not an easy task: prior work that attempted to use SVT for finding heavy nodes~\cite{biswas2026} suffered a subtle bug~\cite{lumberjack2026}; correcting this required extra privacy noise and a lengthy proof~\cite{biswas2026}. 
In contrast, our contribution is a clean algorithm that is similar in shape to the state-of-the-art~\cite{lumberjack2026} while achieving tighter asymptotic bounds. 
We count the clarity of the algorithm and its accompanying proofs as a major part of our contribution, allowing future work to adopt this framework and apply it to related private search problems. 

%% file: 2-formulation.tex
\section{Problem Formulation}\label{sec:problem}

\subsection{Heavy Node Classification}\label{sec:problem-tree}

\vspace{0pt}
\noindent\textbf{Tree structure and input.} Let $T$ be a rooted tree of height $h$ and arity $k$. Let nodes/vertices be $V$ and leaves be $X$. The tree structure and the ordering of its nodes are public. A private database $D$ is a finite multiset of records in $X$. Each record contributes to the count of a leaf node. We write $n := \lvert D \rvert$ as the size of $D$.
For any leaf node $v\in X$, we write its count as
\begin{align}\label{eq:count-leaf}
\textstyle
    c_v(D)=|\{z\in D:z=v\}|, \text{with } n = \sum_{v\in X}c_v(D)
\end{align}
For any node $u\in V$, we write $\Desc(u)$ and $\Anc(u)$ as its descendants and ancestors, respectively. Both $\Desc(u)$ and $\Anc(u)$ include $u$ itself. \emph{Proper} ancestors and descendants are expressed as $\Anc(u)\setminus \{u\}$ and $\Desc(u)\setminus \{u\}$, respectively.

The total count (or weight) of node $u$ is 
\begin{align}\label{eq:count-node}
\textstyle
 c_u(D)=\sum_{v\in X\cap\Desc(u)} c_v(D).
\end{align}
If $v\in\Desc(u)$, then $c_v(D)\le c_u(D)$.

The tree height $h$ is the maximum number of nodes on a root-to-leaf path. We focus on deep trees and assume $h\ge 3$ throughout. The depth of a node starts at zero, so the deepest leaves have depth $h-1$. All logarithms are natural unless a base is explicitly indicated. \looseness=-1

\vspace{0pt}
\noindent\textbf{Classification problem.} Fix a public threshold $\tau\geq 0$. On any given input database $D$, the goal is to classify each node as \emph{Heavy} or \emph{Light} according to whether its count exceeds $\tau$.
We allow an additive error margin of $E$ around the threshold while respecting record-level privacy (defined next). In particular, the following must hold simultaneously for every node $u\in V$, except with error probability at most $\beta$:\\
%\setlength{\itemsep}{0ex}\setlength{\topsep}{0ex}
%\item 
\hspace*{2mm}$\bullet$
If $c_u(D)\ge \tau + E$, we classify node $u$ as \emph{Heavy}. \\
%\item 
\hspace*{2mm}$\bullet$
If $c_u(D)< \tau - E$, we classify node $u$ as \emph{Light}. \\
%\item  
\hspace*{2mm}$\bullet$
If $c_u(D)\in [\tau-E,\tau+E)$, $u$ may receive either label.%

%\end{itemize}

Formally, we want to achieve the following:

\begin{definition}[Classification guarantee]\label{def:utility}
Fix $\tau\ge 0$, $E\ge 0$, and $\beta\in[0,1)$.
A randomized algorithm $\mathcal A$ satisfies the $(E,\beta)$ classification guarantee at threshold $\tau$ if on every input database $D$ it returns an ancestor-closed set
$\widehat H\subseteq V$ such that
\begin{align}\label{eq:utility}
\Pr\!\big[
  &\{u\in V:c_u(D)\ge\tau+E\}
  \subseteq \widehat H \land \\\nonumber
  & \quad\widehat H\subseteq
  \{u\in V:c_u(D)\ge\tau-E\}
\big]
\ge 1-\beta,
\end{align}
where the probability is over the internal randomness of $\mathcal A$. 
A deterministic error margin means $\beta=0$.
\end{definition}

\subsection{Differential Privacy}\label{sec:problem-dp}

\begin{definition}[Add/remove adjacency]
Databases $D,D'$ are neighbors, written $D\sim D'$, when one is obtained from the other by adding one leaf record. 
\end{definition}
Consider any leaf record $z\in X$ and neighboring pairs $D=D'\uplus\{z\}$. Here $\uplus\{z\}$ means adding one occurrence of $z$. Let $P_z$ be the root-to-leaf path of $z$, then
\begin{equation}\label{eq:adjcounts}
 c_{u}(D)-c_{u}(D')=\ind\{u\in P_z\}.
\end{equation}
Namely, the extra record $z$ influences only the counts of the nodes that lie on $P_z$; and its influence is exactly~$1$.

\begin{definition}[Sensitivity~\citep{tcc06}]
The sensitivity of function $f$ is $ \textstyle
        S(f):=\max_{D\sim D'} \|f(D)-f(D')\|,$
where $\|\cdot\|$ is some norm of interest, e.g., the $\mathcal{L}_1$ norm.
\end{definition}

\begin{definition}[$(\eps,\delta)$-differential privacy~\citep{eps-delta-indistinguish,tcc06}]
A randomized algorithm $\mathcal A$ is $(\eps,\delta)$-DP if, for all $D$ and $D'$ s.t. $D\sim D'$ and every measurable output event $S$,
\begin{equation}\label{eq:DP}
 \Prb[\mathcal A(D)\in S]\le e^{\eps}\Prb[\mathcal A(D')\in S]+\delta.
\end{equation}
\end{definition}

% \vspace{0pt}
% \noindent\textbf{Benchmark for comparisons.}
% We compare algorithms at the same privacy parameters $(\varepsilon,\delta)$, tree parameters, and prescribed simultaneous failure bound $\beta$, and with the same public mass bound $M$ when applicable. In the additive-margin formulation, we fix the threshold $\tau$ and compare the \emph{smallest} error margin $E$. 
% For comparison with the relative formulation~\cite{ghazi2023}, we instead fix a relative tolerance $\gamma\in(0,1)$ and compare the smallest sufficient threshold $\tau$ for which $E\le\gamma\tau$.

\begin{definition}[Laplace distribution]
A random variable $Z$ has a zero-mean Laplace distribution with scale $b>0$, written $Z\sim \mathrm{Lap}(b)$, if its probability density function~(pdf) is $
p_Z(z)=\frac{1}{2b}\exp\!\left(-\frac{|z|}{b}\right)$, for $z\in\mathbb{R}$.
\end{definition}

\begin{lemma}[Laplace mechanism]\label{lem:lap-mechanism}
Let $f$ be a counting query with sensitivity $1$. For $\varepsilon>0$, define the Laplace mechanism for  input $D$ as
$\mathcal A(D) = f(D) + Z\,\textrm{ with }\, Z \sim \Lap\!\left(\frac{1}{\varepsilon}\right)$. Then mechanism $\mathcal A$ satisfies $(\varepsilon,0)$-DP.
\end{lemma}

% \subsection{Target Charging Technique}
% \label{sec:problem-tct-background}

\vspace{0pt}
\noindent\textbf{SVT and TCT.} 
The \emph{sparse vector technique} (SVT) is a fundamental tool for
answering adaptively chosen threshold queries under differential
privacy~\cite{svtrevisited,svt-app-10.1145/2488608.2488651,svt-app-mw-linear,svt-app-pmlr-v117-nandi20a,svt-app-pmlr-v83-nissim18a,svt-app-valid,svt-origin,understandingsvt,svt-app-dp-level-erm,svt-app-pmlr-v65-feldman17a}. It is useful when an analyst must examine many queries but expects only a small number of positive responses (in our case, nodes having count $\ge\tau$). 
Under linear composition or advanced composition~\cite{dwork:roth:14,pmlr-v37-kairouz15,boosting-and-dp}, posing $K$ queries would impose a $\Theta(K)$ or $\Theta(\sqrt{K})$ multiplicative factor on the final $\varepsilon$ parameter. \looseness=-1

As a result, more queries often lead to larger errors due to larger DP noise magnitude. 
In contrast, the privacy cost of SVT can be $O(1)$. 
SVT's basic \texttt{AboveThreshold} primitive compares noisy query values against a noisy threshold and stops at the first positive response; 
more general SVT procedures allow a prescribed number of positive responses. 
The privacy bound is governed by the number of positive responses rather than the total number of queries~\cite{svt-origin,dwork:roth:14,svt-app-10.1145/2488608.2488651,svt-app-valid,svt-app-mw-linear,svt-app-dp-level-erm,svt-app-pmlr-v65-feldman17a}.

The \emph{target-charging technique} (TCT)~\cite{cohenlyu2023} extends SVT-style accounting from threshold tests to general DP mechanisms. 
Each run is associated with a \emph{target}, a subset of possible outputs from a DP mechanism. A \emph{target hit} occurs when the released output belongs to that subset. Roughly speaking, the probability mass of the target decomposition and the observed number of target hits control the overall privacy cost. Similar to SVT, this cost can therefore be independent of the overall number of queries to the input database.
For threshold testing over adaptive queries, 
TCT injects fresh noise into each query answer to compare it with a public pre-fixed (noiseless) threshold. 

\begin{definition}[$\lambda$-Indistinguishability]
For probability distributions $P_0,P_1$ on an outcome space
$\mathcal Y$, we say $P_0$ and $ P_1$ are $\lambda$-indistinguishable (written as $P_0\approx_\lambda P_1$) if, 
for every measurable $S\subseteq\mathcal Y$,
\begin{align}
  P_0(S)\le e^\lambda P_1(S)
  \qquad\text{and}\qquad
  P_1(S)\le e^\lambda P_0(S).
\end{align}
\end{definition}

\begin{definition}[$q$-target
  {\cite{cohenlyu2023}}]
\label{def:tct-general-target}
Let $q\in(0,1]$, $\lambda>0$, and $\kappa\in[0,1)$.
A measurable set $\top\subseteq\mathcal Y$ is a \emph{$q$-target} for a pair of probability distributions $(P_0,P_1)$ with parameters $(\lambda,\kappa)$ if there exist $p\in[0,1]$ and 
probability distributions $C,N_0,N_1,F_0,F_1$ on $\mathcal Y$ such that, for
$b\in\{0,1\}$,
\begin{align}
  P_b=(1-\kappa)\bigl(pC+(1-p)N_b\bigr)+\kappa F_b, \quad\text{ and }
  \label{eq:tct-general-decomposition}
\end{align}
\begin{align} \textstyle
  N_0\approx_\lambda N_1,
  \qquad
  \min_{b\in\{0,1\}}N_b(\top)\ge q.
  \label{eq:tct-general-target-conditions}
\end{align}
We say the set $\top$ is a $q$-target of a randomized mechanism $\mathcal A$
with parameters $(\lambda,\kappa)$ if this condition holds for
$(\mathcal A(D),\mathcal A(D'))$ for every neighboring pair $D\sim D'$.
\end{definition}

Sampling from $C$ does not require knowing the bit $b$ selecting the input. Sampling from $N_b$ is a $(\lambda,0)$-DP access to the input and hits the target with probability at least $q$.
The unrestricted component $F_b$ contributes to the exceptional probability $\kappa$. 
This simulation functions primarily as an analysis device and need not be explicitly implemented by the release mechanism.

%% file: 3-related-works.tex
\section{Heavy Node Classification}\label{sec:literature}

Prior work relevant to the private heavy node classification problem has fallen into two categories: histogram-based and threshold-testing.  
Part of our contribution is to unite these previously separate approaches. 

\noindent\textbf{Histogram-based counting}.
The  naive \emph{full-tree} baseline yields additive errors in $\Omega(\operatorname{poly}({h}))$ for heavy node classification. 
The naive Gaussian mechanism~\cite{dwork:roth:14} publishes a histogram of the node counts at every level in the tree, and labels a node heavy exactly when its noisy count is at least $\tau$. 
For a hierarchy of height $h$, the $\mathcal{L}_2$ sensitivity is $\sqrt{h}$.
Adding independent Gaussian noise of standard deviation $O(\sqrt{h}/\eps)$ to each node count results in an additive error margin of $O_{\eps,\delta}(h)$ for classification (see Appendix~\ref{sec:related-naive-baselines}). 

Lebeda et al.~\cite{lumberjack2026} consider more sophisticated histogram approach. Their algorithm, named Lumberjack, queries middle layers and recurses down the tree if the noisy count is heavy, and up the tree if it is light. 
This is conceptually equivalent to performing (parallel) binary searches on a tree.
The result is to publish a collection of histograms of node counts with bins of varying widths. 
Each leaf is counted towards $O(\log h)$ histograms, resulting in $O(\sqrt{\log h})$ error under Gaussian noise. 
For $0<\eps<1$, \cite[Corollaries 17--18]{lumberjack2026} Lumberjack achieves absolute additive error $E_{\mathrm{LJ}}$, as
\begin{align}
\textstyle
 E_{\mathrm{LJ}}
 &=O\left(\frac{1}{\varepsilon}\sqrt{\log(1/\delta)\log h\log\log h}\right).
\label{eq:ljerror}
\end{align}
This error margin is deterministic; namely, it holds with probability $1$. 

\noindent\textbf{Threshold-testing.}
The second line of work uses SVT-like mechanisms to test which nodes' weight is above the threshold, by traversing the hierarchy. 
\cite[Algorithm 1]{ghazi2023} walks the tree from the deepest level toward the root, maintaining the unclassified nodes $S_i$ at each level $i$.
It uses an SVT procedure to find the (disjoint) nodes in a level that are above the (noisy) threshold. 
Heavy nodes and all their ancestors are removed, and the search proceeds to higher levels. 
The number of positive answers $c$ is bounded by $O(M/\tau)$, where $M$ is an upper bound on the total count $n$. 
The resulting additive error margin $E_G$ is given by 
\begin{align}
\textstyle
E_G= \frac{2M}{\eps\tau}\max\left\{48\ln\frac{2h}{\beta},\;
 6\ln\left(1+\frac{e^{\eps/2}-1}{\delta}\right)\right\}.
\end{align}

A related but more challenging task is to find hierarchical heavy hitters. 
The (corrected) algorithm of Biswas et al.~\cite{biswas2026} also utilizes the idea of SVT and bottom-up traversal. They obtain an additive error margin of
\begin{align}
\hspace{-1em} \textstyle
 E_B=O\!\left(
 \frac{\log(1/\delta)}{\eps}
  \log\frac{\log(1/\delta)}{\eps}\left[\log(1/\delta)+\log(h/\beta)\right]
\right).
\end{align}
In summary,~\cite{ghazi2023,biswas2026} give threshold-based algorithms for heavy node classification with absolute additive error margins in $\Omega_{\varepsilon,\delta}(\log h)$ while the histogram-based method of~\cite{lumberjack2026} achieves error of $\Omega_{\varepsilon,\delta}(\sqrt{\log h})$. 
In this work, we converge these paradigms in order to reduce the error for node classification to $O_{\varepsilon,\delta}(\log\log h)$.

Our problem is also related other fundamental tasks in differential privacy, such as heavy hitters and counting for non-hierarchical data~\cite{zhu2020federated,blocki2023differentially,ldp-hh-wang,federated-hh-ldp,differentially-private-heavy-fa,wu2022heavyhitters,practical-ldp-hh}, histogram release~\cite{sandtdp,dp-sparse-linear,lebeda_et_al:LIPIcs.FORC.2025.23,korolova2009,dp-sparse-data,instance-mean}, quantile estimation and learning thresholds~\cite{DP_unbounded,pmlr-v139-gillenwater21a,alabi2023bounded,imola2025differentially,optimal-dp-thresholds,bnsv2015,DBLP:conf/icml/KaplanSS22,DBLP:conf/colt/KaplanLMNS20}, and counting queries on trees and multi-dimensional spaces~\cite{Xiao2012DPCubeDP,rectangle-queries,ghazi2023,zhang2016privtree,dp-spatial,hierarchical-spatial-ldp,chang2021oneround}. 
While some techniques could also be applicable to our setting, we aware of no algorithm that achieves a simultaneous error for all nodes with doubly-logarithmic dependence on tree height $h$ and no dependence on input database size $n$.

%% file: 4-algo.tex
\section{Our Algorithm}\label{sec:algorithm}

\subsection{Building Blocks}
\label{subsec:buildingblocks}
\vspace{0pt}
\noindent\textbf{Finite-horizon, one-hit TCT.}
We present the following specialization of the target-charging simulation of \cite{cohenlyu2023}: the interaction stops at its first target hit and makes at most $K$ calls.

\begin{lemma}[Finite-horizon one-hit TCT]
\label{lem:tct-one-hit}
Fix an integer $K\ge 1$ and parameters $\lambda>0$, $q\in(0,1]$, and $\kappa_0\in[0,1)$.
Let $\mathcal A$ be an algorithm that accesses its input database $D$ only through adaptive calls to randomized mechanisms.
At round $i$, $\mathcal A$ selects a mechanism $\mathcal M_i$ with finite output space $\mathcal Y_i$ and a target $\top_i\subseteq\mathcal Y_i$, and obtains $Y_i\sim\mathcal M_i(D)$ using fresh randomness.
The selections depend only on public information and previous responses.
For every round and every possible response history, assume that $\top_i$ is a $q$-target of $\mathcal M_i$ with parameters $(\lambda,\kappa_0)$ (Definition~\ref{def:tct-general-target}).
Algorithm $\mathcal A$ stops making calls immediately after its first target hit $Y_i\in\top_i$, or after $K$ calls if no hit occurs. Its output is any post-processing of the response sequence, including its length.
Then, for every integer $1\le r\le K$, $\mathcal A$ is
\begin{align}
\label{eq:finite-one-hit-tct}
  \left(
    r\lambda,\;
    K\kappa_0+\mathbf 1_{\{r<K\}}(1-q)^r
  \right)\text{-DP}.
\end{align}
\end{lemma}
Here $r$ should be understood as describing the budget for draws from $N_b$---it is not an additional stopping rule for the interaction. 
In the pairwise simulation, each draw from $N_b$ costs $\lambda$ privacy budget in $\varepsilon$ and hits its target with conditional probability at least $q$.
Thus, more than $r$ draws requires the first $r$ to miss their targets, an event of probability at most $(1-q)^r$.
This overflow is impossible when $r=K$.
The term $K\kappa_0$ accounts for selecting the exceptional mass component in any of the calls.
The proof appears in Appendix~\ref{sec:app-proof-one-hit}.

\begin{algorithm}[t]
\caption{Count query $\GBT(c;a,b,R,\lambda)$}\label{alg:guard}
\begin{algorithmic}[1]
\Require Original node count $c\in\N$; public thresholds $a,b\in\mathbb{R}$ with $a<b$, guarded margin $R\ge1$ for a deterministic error, local privacy parameter $\lambda>0$
\If{$c\le a-R$} \Comment{deterministic output}
  \State \Return $\Lo$
\ElsIf{$c\ge b+R$} \Comment{deterministic output}
  \State \Return $\Hi$
\Else \Comment{draw fresh DP noise}
  \State Draw fresh independent $Z\sim\Lap(1/\lambda)$
  \If{$c+Z<a$}
    \State \Return $\Lo$
  \ElsIf{$c+Z>b$}
    \State \Return $\Hi$
  \Else
    \State \Return $\Bt$
  \EndIf
\EndIf
\end{algorithmic}
\end{algorithm}

\vspace{0pt}
\noindent\textbf{Count query.} Given threshold $\tau$, choose parameters $0<\lambda\le1$ and $R\ge1$, and define
\begin{align}\label{eq:primitiveparams}
\textstyle
\hspace{-1em} w=\frac1\lambda,\hspace{1em} a=\tau-w,\hspace{1em}
b=\tau+w,
% \hspace{1em} E=w+R.
\end{align}
Algorithm~\ref{alg:guard} is the building block for our algorithm. 
Given a (true) node count $c$ to inspect, we call the routine $\GBT(c;a,b,R,\lambda)$, and return $\Lo$ or $\Hi$, indicating whether the noisy count is smaller or larger than the public thresholds $a$ and $b$, respectively. If the returned value is $\Bt$, i.e., `between', that means the value is judged to be close to the thresholds $a$ and $b$. 
The semantics of a $\Bt$ response is that the node count is on the borderline: it is small enough to be treated as effectively light ($\Lo$), and simultaneously large enough to be treated as effectively heavy ($\Hi$).  
Our algorithm will use this duality to treat any discovered $\Bt$ labeled nodes as marking the boundary between heavy ancestors and light descendants. 
Here, the parameter $R$ can be thought of as the maximum deterministic error margin, which is used to avoid a union bound over all queries. 
$\lambda$ can be viewed as the local privacy parameter (not the final $\eps$). 

\FloatBarrier
\noindent\textbf{Public traversal ordering.} On any tree $T$ of height $h$, we call \textsc{DepO}$(0,h-1)$ (Algorithm~\ref{alg:public-order}) to determine the traversal ordering for the different levels of the public tree. 
It can be understood in terms of binary search: first the middle layer is visited, then recursively the lower layers, followed by the higher layers. 
This ordering is public; it does not rely on information from input database. For example, on $h=15$, applying the definition gives the ordering $[7, 11, 13, 14, 12, 9, 10, 8, 3, 5, 6, 4, 1, 2, 0]$.

\begin{algorithm}[t]
\caption{Depth Traversal Ordering \textsc{DepO}$(i,j)$}\label{alg:public-order}
\begin{algorithmic}[1]
 \If{$i>j$} \State \Return the empty list \EndIf
 \State $m\gets i+\lfloor(j-i+1)/2\rfloor$
 \State \Return $[m]\,\Vert\,\Call{DepO}{m+1,j}\,\Vert\,\Call{DepO}{i,m-1}$
\end{algorithmic}
\end{algorithm}

\vspace{0pt}
\noindent\textbf{Label propagation.} 
For node $u$ with count $c$, if $\GBT(c;a,b,R,\lambda)$ returns $\Lo$, we classify all its descendants as well as $u$ as $\Lo$, by monotonicity. 
Similarly, if $\GBT(c;a,b,R,\lambda)$ returns $\Hi$, we classify all its ancestors as well as $u$ as $\Hi$; if $\GBT(c;a,b,R,\lambda)$ returns $\Bt$, we leverage its boundary nature to classify all its descendants as well as $u$ as $\Lo$ and all its ancestors (excluding $u$) as $\Hi$. 
We refer to these routines as $\mathsf{MarkUp}(u)$ and $\mathsf{MarkDown}(u)$ for propagating $\Hi$ and $\Lo$ labels, respectively. 
Other nodes in the tree that don't overlap with $u$'s records are unaffected by the marking procedure, since their counts are not directly comparable with $u$'s. For completeness, 
Algorithm~\ref{alg:marking-routines} in Appendix~\ref{sec:app-marking} formalizes this simple rule of label propagation, originating from node $u$.

\subsection{The \textsc{BetweenCut} Algorithm}

\begin{algorithm}[t]
\caption{Node classifier $\mathsf{BetweenCut}(T,D,\tau,\lambda,R)$}\label{alg:tree}
\begin{algorithmic}[1]
\Require Public tree $T$ of height $h$; private leaf dataset $D$, public parameters $\tau,\lambda,R$.
\State $w\gets1/\lambda$; $a\gets\tau-w$; $b\gets\tau+w$
\State Compute all original counts $c_u(D)$
\State Initialize $\mathrm{label}[u]\gets\Unk$ for all $u\in V$
\For{$d$ in \Call{DepO}{$0,h-1$}}
 \For{nodes $u$ of depth $d$
 % , in the fixed public left-to-right order
 }
  \If{$\mathrm{label}[u]\ne\Unk$} \State \textbf{continue} \EndIf
  \State $s\gets\GBT(c_u(D);a,b,R,\lambda)$
  \If{$s=\Lo$}
   \State \Call{MarkDown}{$u$}
  \ElsIf{$s=\Hi$}
   \State \Call{MarkUp}{$u$}
  \Else % \Comment{\parbox[t]{0.58\linewidth}{$s=\Bt$: remove every root-to-leaf path through $u$}}
   \State \Call{MarkUp}{$\parent(u)$}
   \State \Call{MarkDown}{$u$}
  \EndIf
 \EndFor
\EndFor
\State \Return $\Hhat=\{u\in V:\mathrm{label}[u]=\Hi\}$
% \Statex \Comment{A public wrapper restricts this set to original nodes when padding was used.}
\end{algorithmic}
\end{algorithm}
We present $\mathsf{BetweenCut}$, listed as in Algorithm~\ref{alg:tree}. 
The overarching idea is to turn a node $u$'s answer near the given threshold $\tau$ (in particular, a `between' answer from Algorithm~\ref{alg:guard}) into a stopping rule that resolves all affected paths through that node, by propagating labels to $u$'s ancestors and descendants (using Algorithm~\ref{alg:marking-routines}). 
Overall, on each root-to-leaf path, our algorithm makes at most $\lceil\log_2 (h+1)\rceil$ query calls of $\GBT(c_u(D);a,b,R,\lambda)$ through some node $u$ in that path. 
This bound follows from the midpoint-first traversal and label propagation. Within any recursive interval of levels, the middle level is processed first. Consider node $u$ on path $P_z$. A query at $u$ labels its descendants Light for an $\Lo$ response, or its ancestors Heavy for an $\Hi$ response; a $\Bt$ response resolves the entire path. After the middle level is processed, all remaining queries on $P_z$ within that interval lie in at most one of its two subintervals.\looseness=-1

\noindent\textbf{Remark on consistency and termination.}
After each node is processed, the Heavy set (labeled as $\Hi$) is ancestor-closed and the Light set (labeled as $\Lo$) is descendant-closed. The propagation rule of labels also cannot conflict with existing labels; hence, each node is labeled at most once and its label remains unchanged afterwards. Every node is eventually labeled, so the algorithm terminates with a complete classification. We formalize and prove these properties in
Appendix~\ref{app:algorithm-properties}. 

The next two lemmas establish the query bound and stopping property crucial to the privacy analysis.

\begin{lemma}[Logarithmically many queries on each root-to-leaf path]
\label{lem:path-query-bound}
For every database $D$ and every execution of Algorithm~\ref{alg:tree}, let $\mathcal Q$ be the set of nodes at which Algorithm~\ref{alg:guard} is invoked.
Then every root-to-leaf path $P$ of the public tree satisfies
\begin{align}
    |\mathcal Q\cap P|
    \le \lceil\log_2(h+1)\rceil.
\end{align}
\end{lemma}
% In particular, Lemma~\ref{lem:path-query-bound} applies to $P_z$ even
% when $z\notin D$, as needed when comparing $D$ with
% $D\uplus\{z\}$.

\begin{lemma}[At most one middle target on each root-to-leaf path]
\label{lem:one-middle-target}
For every execution of Algorithm~\ref{alg:tree},
a $\Bt$ response
at a node $u$ prevents all subsequent queries on every
root-to-leaf path containing $u$.
Consequently, if $\mathcal{B}$ is the set of queried nodes
whose responses are $\Bt$, then every root-to-leaf
path $P$ satisfies $|\mathcal{B}\cap P|\le 1$.
\end{lemma}
If we can prove that $\{\Bt\}$ is a $q$-target for Algorithm~\ref{alg:tree} (recall Def.~\ref{def:tct-general-target}), then Lemmas~\ref{lem:path-query-bound} and~\ref{lem:one-middle-target} enable us to utilize Lemma~\ref{lem:tct-one-hit}---the bespoke TCT analysis with at most $\lceil \log_2(h+1) \rceil$ calls and at most $1$ target hit---to analyze the privacy of Algorithm~\ref{alg:tree}.

%% file: 6-analysis-privacy-trade-off.tex
\section{Privacy and Utility}\label{sec:analysis}

In this section, we present the privacy and utility guarantees of Algorithm~\ref{alg:tree}. The formal proofs are deferred to Appendix~\ref{sec:app-transcript-privacy} and~\ref{sec:app-proof-utility}.
Studying Algorithm~\ref{alg:tree} and its properties shown above, it is clear that it makes only $O(\log h)$ queries that depend on any input item. 
However, each such query materializes a noisy count, so it might appear that the privacy cost should behave like $O(\log h)$, or maybe $O(\sqrt{\log h})$ at best, due to composition bounds. 
This would give a result similar to the Lumberjack algorithm. 
Our approach to more tightly bounding this cost is based on the careful choice of query that is posed.  
The insight is that it is only the ``boundary inputs'' that incur a genuine privacy cost. 
Essentially, if the count of a node is far below the threshold $\tau$, then we can reveal this fact without any impact on privacy; symmetrically, if the count is sufficiently far above the threshold, revealing this also does not cost privacy. 
It is only the inputs that reside close to $\tau$ that risk private information. 
However, for such inputs, we are most likely to receive a $\Bt$ output for them, which stems the leakage of further private information. 
In particular, we can bound the probability that boundary inputs along a path repeatedly fail to yield a $\Bt$ output. 
This fits neatly into the specialized version of TCT that we introduced in Section~\ref{subsec:buildingblocks}.

\subsection{Privacy Analysis}

For Algorithm~\ref{alg:guard}, we instantiate the $q$-target decomposition (recall Definition~\ref{def:tct-general-target}) for target outcome $\{\Bt\}$.
Informally, we decompose the probability that our algorithm outputs $\Bt$ into component probabilities. 

\begin{definition}[Target decomposition for outcome $\{\Bt\}$]\label{def:target}
The target $\{\Bt\}$ is a $q$-target for a pair $(P_0,P_1)$ with parameters $(\lambda,\kappa)$ if there is $p\in[0,1]$ and probability distributions $C,N_0,N_1,F_0,F_1$ s.t., for $b\in\{0,1\}$,
\begin{align}\label{eq:targetdecomp}
 P_b=(1-\kappa)\bigl(pC+(1-p)N_b\bigr)+\kappa F_b,
\end{align}
where $N_0\approx_\lambda N_1$ and $\min_b N_b(\Bt)\ge q$.
\end{definition}
Note that the common component $C$ costs no privacy for the fixed pair. The residual component $N_b$ is a pure-private access and has a substantial chance to hit the target (i.e., with probability $\ge q$). The component $F_b$ is paid for in final $\delta$. 
This decomposition is for analysis purpose only; it is never computed by the released mechanism and need not be common to different neighboring pairs. Next, we bound the parameters, exceptional mass $\kappa$ and target hit probability $q$.

\vspace{0pt}
\noindent\textbf{Target decomposition under pure-DP.} We warm up with a pure-DP version of Algorithm~\ref{alg:guard}. Let $\BT(c; a, b, \lambda)$ denote \cref{alg:guard} without either deterministic guard, listed as in Algorithm~\ref{alg:no-guard}. Let $P_c$ be its law on $\{\Lo,\Bt,\Hi\}$. Applying the ``Effectiveness of the between target'' Lemma~\cite[Lemma G.1]{cohenlyu2023}, we obtain the following. 
For completeness, we provide the full self-contained proof in Appendix~\ref{sec:app-transcript-privacy}.

\begin{algorithm}[t]
\caption{No-guard count query $\BT(c;a,b,\lambda)$}\label{alg:no-guard}
\begin{algorithmic}[1]
\Require Original node count $c\in\N$; public thresholds $a,b\in\mathbb{R}$ with $a<b$, local privacy parameter $\lambda>0$
  \State Draw fresh independent $Z\sim\Lap(1/\lambda)$
  \If{$c+Z<a$}
    \State \Return $\Lo$
  \ElsIf{$c+Z>b$}
    \State \Return $\Hi$
  \Else
    \State \Return $\Bt$
  \EndIf
\end{algorithmic}
\end{algorithm}

\begin{lemma}[Target decomposition for $\BT$]\label{lem:BTtarget}
Fix $\lambda>0$ and public thresholds $a<b$ satisfying
$b-a=2/\lambda$. For each count $c\in\mathbb Z_{\ge0}$, let $P_c$
denote the output distribution of the unguarded mechanism
$\operatorname{BT}(c;a,b,\lambda)$ on
$\mathcal Y:=\{\mathsf L,\mathsf B,\mathsf H\}$. For any $c_0,c_1\in\mathbb Z_{\ge0}$ with $|c_1-c_0|\le1$,
the singleton target $\{\mathsf B\}$ is a $\qstar$-target for the pair
$(P_{c_0},P_{c_1})$ with parameters $(\lambda,0)$, where
\begin{align}\label{eq:qexact}
%\textstyle
  \qstar
  =
  \frac{\bigl(1-e^{-(b-a)\lambda}\bigr)(e^\lambda-1)}
       {e^{2\lambda}-1}
  =
  \frac{1-e^{-2}}{1+e^\lambda}.
\end{align}
Specifically, there exist $p\in[0,1]$ and probability distributions
$C,N_0,N_1$ on $\mathcal Y$ such that
\begin{align}
  &P_{c_i}= p C + (1 - p) N_i,
  \qquad i\in\{0,1\},\\
&\text{and }\quad
  N_0\approx_\lambda N_1,
  \qquad
  \min_{i\in\{0,1\}}N_i(\mathsf B)\ge\qstar.
\end{align}
% The witnesses $p,C,N_0,N_1$ may depend on the pair $(c_0,c_1)$.
% If, additionally, $0<\lambda\le1$, then $\qstar>1/5$.
\end{lemma}

\vspace{0pt}
\noindent\textbf{Target decomposition of Algorithm~\ref{alg:guard}.} 
We now analyze the full version of the algorithm which includes the deterministic outputs for inputs that are truly far from the threshold.  
This requires an approximate DP analysis. 
Restoring these deterministic guards recovers Algorithm~\ref{alg:guard}.

\begin{lemma}[Target decomposition for $\GBT$]\label{lem:guards}
Fix public parameters $0<\lambda\le1$, $R\ge1$, and $b-a=2/\lambda$.
For each count $c\in\mathbb Z_{\ge0}$, let $Q_c$ denote the output
distribution of the guarded mechanism $\GBT(c;a,b,R,\lambda)$ on
$\mathcal Y:=\{\mathsf L,\mathsf B,\mathsf H\}$.

For any $c_0,c_1\in\mathbb Z_{\ge0}$ with $|c_1-c_0|\le1$, the singleton target $\{\mathsf B\}$ is a $\qstar$-target for the pair $(Q_{c_0},Q_{c_1})$ with parameters $(\lambda,\kappa)$ for some $\kappa\in[0,\delta_0]$, where $\qstar$ is defined in Eq.~\eqref{eq:qexact} and
\begin{align}\label{eq:localdelta}
\textstyle
  \delta_0=\frac12e^{-\lambda(R-1)}.
\end{align}
Specifically, there exist $\kappa\in[0,\delta_0]$, $p\in[0,1]$,
and probability distributions $C,N_0,N_1,F_0,F_1$ on $\mathcal Y$
such that
\begin{align}
\hspace{-1em}  Q_{c_i}
  =
  (1-\kappa)\bigl(pC+(1-p)N_i\bigr)+\kappa F_i, i\in\{0,1\},  
\end{align}
\begin{align}\textstyle
\text{ and }  N_0\approx_\lambda N_1,
  \hspace{1em}
  \min_{i\in\{0,1\}}N_i(\mathsf B)\ge\qstar.
\end{align}

\end{lemma}

The guard test ensures that any count that is far above $a$ or below $b$ is classified correctly with probability $1$. 
This property does not come for free in the analysis. 
We pay with an additional $\delta_0$ in privacy. 
Then the proof of Lemma~\ref{lem:guards} is rather intuitive. There are three cases to consider: i) both inputs $c_0$ and $c_1$ are below $a-R$ or above $b+R$ ii) both $c_0$ and $c_1$ are within $[a-R,b+R]$, iii) one of $c_0$ or $c_1$ is below $a-R$ or above $b+R$ while the other is within $[a-R,b+R]$. 
Note that the case where one of $c_0$ and $c_1$ appears $(-\infty,a-R)$ and the other appears in $(b+R,\infty)$ is impossible, since $R\ge 1$ and by assumption we have $|c_0-c_1|\le 1$. 
Case i) is trivial to handle; case ii) degenerates to the no-guard Algorithm~\ref{alg:no-guard}, which is handled by Lemma~\ref{lem:BTtarget}. 
For case iii), we bound the error probability $\delta_0$ that the answers on $c_0$ and $c_1$ disagree, which is the probability that $Z\ge R-1$ for $Z\sim \Lap(1/\lambda)$. 

For neighboring databases differing by one leaf record $z$, only counts on the root-to-leaf path $P_z$ change. Algorithm~\ref{alg:tree} makes at most
$\lceil\log_2(h+1)\rceil$ queries on this path and makes
no further query there after a $\Bt$ response (namely, a target hit).
Fixing the off-path responses, we can apply the bespoke
finite-horizon one-hit TCT bound in Lemma~\ref{lem:tct-one-hit} to the affected queries on $P_z$, and then reconstructs the complete transcript for Algorithm~\ref{alg:tree}. 
This gives the following privacy guarantee. The full proof is
deferred to Appendix~\ref{sec:app-transcript-privacy}.

\begin{theorem}[Privacy of Algorithm~\ref{alg:tree}]
\label{thm:guarded-tree-privacy}
Let $T$ be a finite public rooted tree of height $h\ge 1$,
with fixed public traversal orders.
Fix public parameters $\tau\ge 0$, $0<\lambda\le 1$,
and $R\ge 1$.
Run Algorithm~\ref{alg:tree} with thresholds
$a:=\tau-1/\lambda$ and $b:=\tau+1/\lambda$, using fresh
independent noise in every randomized invocation of $\GBT$.
Write
$L:=\lceil\log_2(h+1)\rceil, 
  \qstar:=\frac{1-e^{-2}}{1+e^\lambda}.
$ For every integer $1\le r\le L$,
Algorithm~\ref{alg:tree} is
\begin{align}
\label{eq:dense-privacy-curve}\textstyle
  \left(
    r\lambda,\;
    \frac{L}{2}e^{-\lambda(R-1)}
    +\mathbf{1}_{\{r<L\}}(1-\qstar)^r
  \right)\text{-DP}.
\end{align}
\end{theorem}
Here the integer $r\in\{1,\ldots,L\}$ is referred to as the \emph{private-call cutoff}. 
It is used for analysis purpose only. 
We charge at most $r$ calls to the data holder (i.e., the mechanism runs on $N_b$ in Eq.~\eqref{eq:targetdecomp}). 
This leads to a factor of $r$ in the multiplicative privacy parameter (accounted in the final $\varepsilon$) and we treat the event of more than $r$ such calls as an additive privacy failure (accounted in the final $\delta$). 
The conclusion holds for every such $r\in [1, L] = [1, \lceil\log_2 (h+1)\rceil]$. 

We now choose the public parameters so that Algortihm~\ref{alg:tree} is calibrated to $(\varepsilon,\delta)$-DP.

\begin{theorem}[Privacy with explicit parameters]
\label{thm:calibrated-privacy}
Let $T$ be a finite public rooted tree of height $h\ge 1$,
with fixed public traversal orders, and let $\tau\ge 0$
be a public count threshold.
Fix privacy parameters $0<\varepsilon\le 1$ and
$0<\delta<1/10$.
Set $L:=\lceil\log_2(h+1)\rceil$ and $r_0
    := \left\lceil
        \frac{\ln(2/\delta)}{\ln(5/4)}
       \right\rceil$. Choose 
       $r:=\min\left\{L,\; r_0\right\}$ 
and define
\begin{align}
\label{eq:privacy-calibration-parameters}
%\textstyle
  \lambda:=\frac{\varepsilon}{r},
  \qquad
  w:=\frac{1}{\lambda}=\frac{r}{\varepsilon},
  \qquad
  R:=1+\frac{r}{\varepsilon}\ln\frac{L}{\delta}.
\end{align}
Then Algorithm~\ref{alg:tree} with $\lambda, R$ and thresholds
$a:=\tau-w$ and $b:=\tau+w$ satisfies
$(\varepsilon,\delta)$-DP.
\end{theorem}

\subsection{Utility Analysis}\label{sec:utility}

By Theorem~\ref{thm:calibrated-privacy} and the construction of the count query $\GBT$ (Algorithm~\ref{alg:guard}), we have the following privacy-utility trade-off for Algorithm~\ref{alg:tree}.
\begin{theorem}[Utility guarantee]\label{thm:main}
Fix $h\ge1$, $\tau\ge0$, $0<\eps\le1$, and $0<\delta<1/10$. Let $L$, $r$, $\lambda$, $w$,  and $R$ be defined as in Theorem~\ref{thm:calibrated-privacy}. Then Algorithm~\ref{alg:tree} is $(\eps,\delta)$-DP . In addition, the output labels satisfy the $(E,\beta)$ classification guarantee of Definition~\ref{def:utility} with error probability $\beta=0$ and additive error margin
\begin{align}\textstyle
 E:=w+R=1+\frac{r}{\eps}\left(1+\ln\frac{L}{\delta}\right).\label{eq:mainerror}
\end{align}
\end{theorem}
% This follows from Theorem~\ref{thm:calibrated-privacy} and the fact that the maximum classification error is bounded by $w+R$ with probability $1$. 

Eq.~\eqref{eq:mainerror} satisfies the explicit relaxation
\begin{align}\label{eq:simpleconstant}\textstyle
 E\le1+\frac{5\ln(2/\delta)}{\eps}
       \left(1+\ln\frac{\log_2 h}{\delta}\right)=O_{\varepsilon,\delta}\left(\log\log h\right),
\end{align}
which is doubly logarithmic in $h$ at fixed $\varepsilon$ and $\delta$. 

\vspace{0pt}
\noindent\textbf{Remark on time complexity.} Algorithm~\ref{alg:tree} marks every node in tree $T$. For a tree of arity $k$ and height $h$, the overall time complexity is $O(k^h)$. This cost can be reduced to $O(nh)$, where $n$ is the input database size, utilizing a sparsity-aware implementation of Algorithm~\ref{alg:tree}. The algorithm, its properties, and proofs are deferred to Appendix~\ref{sec:spv4-sparse}; we focus on the privacy-utility trade-off in the main paper.

%% file: 7-comparison-short.tex
% Drop-in subsection. Required packages: amsmath, amssymb, booktabs,
% tabularx, array. Reuses bibliography keys lumberjack2026 and ghazi2023.
\subsection{Discussion and Comparison to Prior Work}
\label{sec:cmp:additive}

\begin{table}[t]
\centering
\caption{Additive error margin comparison at height $h$. An entry $f$ denotes an error margin $E=O(f)$.}
\label{tab:cmp:dependencies}
\small
\setlength{\tabcolsep}{4pt}
\renewcommand{\arraystretch}{1.3}
\begin{tabularx}{\linewidth}{@{}>{\raggedright\arraybackslash}X r@{}}
\toprule
Method & Additive error margin $O(\cdot)$ \\
\midrule

\multicolumn{2}{@{}l@{}}{%
  \textbf{Probabilistic error guarantee} ($\beta>0$)} \\
% \addlinespace[0.5ex]

Ghazi et al.~\cite{ghazi2023}
& $\displaystyle
  \frac{1}{\varepsilon}\frac{M}{\tau}
  \left[
    \log\!\left(1+\frac{\varepsilon}{\delta}\right)+\log\frac{h}{\beta}
  \right]$
\\
\addlinespace[0.5ex]

Biswas et al.~\cite{biswas2026} & $\displaystyle\frac
{\log(\frac1\delta)}{\eps}
  \log\frac{\log(\frac1\delta)}{\eps}\left[\log\left(\frac1\delta\right)+\log\left(\frac{h}{\beta}\right)\right]$\\
\addlinespace[0.5ex]

Lumberjack~\cite{lumberjack2026}
& $\displaystyle
  \frac{1}{\varepsilon}
  \sqrt{
    \log(1/\delta)\,\log h\,
    \log\!\left(\frac{M\log h}{\beta}\right)
  }$
\\

\midrule
\multicolumn{2}{@{}l@{}}{%
  \textbf{Deterministic error guarantee} ($\beta=0$)} \\
% \addlinespace[0.5ex]

Lumberjack~\cite{lumberjack2026}
& $\displaystyle
  \frac{1}{\varepsilon}
  \sqrt{
    \log(1/\delta)\,\log h\,
    \bigl(\log(1/\delta)+\log\log h\bigr)
  }$
\\
% \addlinespace[0.5ex]

Ours
& $\displaystyle
  \frac{1}{\varepsilon}
  \log(1/\delta)
  \bigl(\log(1/\delta)+\log\log h\bigr)$
\\

\bottomrule
\end{tabularx}
\end{table}

Table~\ref{tab:cmp:dependencies} compares \textsc{BetweenCut}, Lumberjack~\cite{lumberjack2026}, 
Ghazi et al.~\cite{ghazi2023} and Biswas et al.~\cite{biswas2026} on one public tree, at the same threshold $\tau$ and under $(\varepsilon,\delta)$-DP.
An additive margin $E$ guarantees that, simultaneously for all nodes, counts at least $\tau+E$ are labeled Heavy and counts below $\tau-E$ are labeled Light, except with probability at most $\beta$.
Here $h$ is the tree height, $n$ is the database size, and $M\ge n$ is a public size bound when required.
We take privacy parameters $0<\varepsilon\le 1$ and $0<\delta<1/10$, and error probability $0<\beta<1/2$; deterministic guarantees are indicated by $\beta=0$. 
For~\cite{ghazi2023}, we convert relative error $\alpha$ to additive by setting $E=\alpha\tau$.

\balance

\vspace{0pt}
\noindent\textbf{Height dependence.}
For fixed privacy parameters, our error margin is $O(\log\log h)$, with error probability $\beta=0$. For fixed $\varepsilon,\delta,\beta>0$, Ghazi et al.'s bound is $O(\frac{M}{\tau}\log h)$; Biswas et al.'s algorithm gives $O(\log h)$. 
Lumberjack reduces each record's participation to $O(\log h)$ comparisons and pays a Gaussian square-root composition factor, giving errors in $O(\sqrt{\log h\,\log\log h})$. 
While the deterministic simultaneous accuracy ($\beta=0$) is already available in Lumberjack, our improvement concerns the exponentially better dependence on $h$.

\vspace{0pt}
\noindent\textbf{On input size $n\le M$.} Our deterministic error guarantee and the deterministic version of Lumberjack do not require an upper bound on the input size $n$. 
In Ghazi et al.~\cite{ghazi2023}, $M$ impacts the utility trade-off through the privacy budget: there can be at most $O(M/\tau)$ positive answers during the SVT invocations. 
In the probabilistic version of Lumberjack~\cite{lumberjack2026}, the $M$ factor is introduced by the union bound over $O(n)$ independent, unbounded noises. \looseness=-1

\vspace{0pt}
\noindent\textbf{On parameter regime $\varepsilon \le 1$.} Our results can be extended to any $\varepsilon > 1$. The $\varepsilon \le 1$ restriction simplifies our parameter calibration rather than being intrinsic to \textsc{BetweenCut}. 
In particular, it ensures that $\lambda=\varepsilon/r\le 1$ in each call of $\GBT$, hence the $q$-target decomposition satisfies $q_\lambda\ge1/5$, so the TCT overflow term is bounded by $(4/5)^r$.
Lumberjack similarly assumes $\varepsilon\le1$ in its closed-form Gaussian calibrations for simplicity---their restriction comes from the classical Gaussian mechanism bound used to obtain a noise scale proportional to $1/\varepsilon$~\cite{dwork:roth:14}. 
Accordingly, \textsc{BetweenCut} can be calibrated for $\varepsilon> 1$, but exploiting the larger privacy budget requires more detailed analysis on the residual-call cutoff threshold. 
We omit this additional analysis for brevity here.

\vspace{0pt}
\noindent\textbf{Relation to lower bounds in~\cite{ghazi2023}.}
Our result is consistent with the lower bounds in~\cite{ghazi2023}. ~\cite[Theorem~5.2]{ghazi2023} establish an
$\Omega_{\varepsilon,\delta}(\sqrt h)$ lower bound for
numerical estimation of all node counts under an additive-only error criterion. Our \textsc{BetweenCut} releases only Heavy/Light
labels at a single public threshold $\tau$, not numerical count
estimates with uniformly bounded additive error. Their pure-DP lower bound~\cite[Theorem~5.1]{ghazi2023} also does not apply to our approximate-DP guarantee. For every fixed multiplicative tolerance $0<\alpha<1$,
a geometric-threshold reduction using \textsc{BetweenCut} achieves $\alpha$-mRMSE $O_{\alpha,\varepsilon,\delta}(\log\log h)$,
improving the height dependence of the $O_{\alpha,\varepsilon,\delta}(\log h)$ bound in~\cite[Theorem~4.1]{ghazi2023}.
The construction and proofs appear in Appendix~\ref{sec:app-extension}.
We leave matching height-dependent lower bounds for approximate-DP
heavy-node classification and additive--multiplicative tree
counting to future work.

%\vspace{0pt}
\noindent\textbf{Relation to Lumberjack's algorithmic design.}
\textsc{BetweenCut} shares Lumberjack's balanced traversal and upward/downward label propagation; the two-sided deterministic guards to avoid union bounding over all node also appear in~\cite{lumberjack2026}. 
Our central innovation is to turn a middle `Between' response into a permanent cut (hence the name \textsc{BetweenCut}): all proper ancestors are marked Heavy, while the queried node and its descendants are marked Light.
This stops every subsequent query on every root-to-leaf path through that node, so each record participates in at most one middle response. 
Constructing the $q$-target decomposition, we obtain a finite-horizon TCT analysis that charges $O(\log(1/\delta))$ calls to the actual private input, with the overflow probability absorbed into $\delta$. 
This avoids accounting for all $O(\log h)$ affected tests, which would lead to an $\Omega(\operatorname{poly}(\log h))$ error as in~\cite{lumberjack2026}. 
This approach allows us to obtain an exponential improvement in height dependence over~\cite{lumberjack2026}.
The compact design of our solution---one public schedule, a three-outcome test, and deterministic label propagation---offers a promising primitive for private hierarchical counting and other structured-data applications.\looseness=-1

%% file: 8-conclusion.tex
\section{Conclusion}\label{sec:conclusion}

We introduce an $(\eps,\delta)$-DP algorithm for heavy node classification on a public tree of height $h$, with an input size-independent additive error margin of $O_{\eps,\delta}(\log\log h)$ that holds simultaneously for all nodes with probability one (Theorem~\ref{thm:main}). 
The key idea is to turn a target hit (middle answer for the BetweenThresholds mechanism) into a final cut: each record can affect at most one target hit, while balanced level search limits its affected queries to O$(\log h)$. 
Our analysis instantiates target-charging technique in this finite-horizon, one-target-hit setting, yielding the explicit privacy guarantee (Lemma~\ref{lem:tct-one-hit}), which may be of independent interest---our formulation is useful when at most one target can be encountered and there are too few target hits to amortize the overhead of the general TCT analysis. A limitation of the current analysis is that the hidden constants are sufficiently large that the point where \textsc{BetweenCut} achieves smallest absolute error is only for extremely deep trees. 
Future directions therefore include tightening the dependence on $\varepsilon$ and $\delta$, determining the optimal height dependence, and investigating downstream applications.\looseness=-1

%% file: AI.tex
% \clearpage
\subsection*{AI use statement}
In this work, we used generative AI tools for checking our proof details, identifying potential errors and ambiguities, deriving constants in the big-oh analyses, and improving the paper writing such as finding grammatical errors and missing references. The research questions and main results are determined and produced by the authors. We have not used generative AI tools for designing the algorithm, and experiments are not applicable to this complexity-focused work. We have reviewed all AI-assisted work: we have confirmed that the computation is correct line by line. We take responsibility for the final content of this work, including all claims and proofs produced with the aid of generative AI.

%% file: App-other-related.tex
\section{Extended Discussion on Related Works}\label{app:related}

\subsection{Naive Baselines for Differentially Private Heavy Nodes}\label{sec:related-naive-baselines}
The following \emph{full-tree} baselines visit every node, adds independent noise to their true node count, and labels them heavy exactly when their noisy count is at least $\tau$. 
Algorithm~\ref{alg:full-tree-baseline} gives a common template for the Laplace and Gaussian variants; their noise distributions are specified below.

\begin{algorithm}[h!]
\caption{Full-tree Classification Baseline}
\label{alg:full-tree-baseline}
\begin{algorithmic}[1]
\Require Public tree $T$ of height $h$; public threshold $\tau$; private database $D$; 
         noise distribution $\mathcal{P}$.
\For{$i=h-1,h-2,\ldots,0$}
    \For{each $u$ at depth $i$}
        \State Draw $Z_u\sim\mathcal{P}$ independently.
        \State $\widetilde{c}_u(D)\gets c_u(D)+Z_u$.
        \If{$\widetilde{c}_u(D)\ge\tau$}
    \State Classify node $u$ as \emph{Heavy}
  \Else
    \State Classify node $u$ as \emph{Light}
  \EndIf
    \EndFor
\EndFor
\end{algorithmic}
\end{algorithm}

\vspace{0pt}
\noindent\textbf{Laplace baseline.}
For $\varepsilon>0$, choose the noise distribution $\mathcal{P}$ to be $\Lap\!\left(\frac{h}{\varepsilon}\right)$. Each node count has sensitivity one. Within each level, the descendant leaf sets are disjoint, so the full count vector for all nodes has an $\ell_1$-sensitivity of $h$. 
The thresholding step is then just post-processing of the DP output~\cite{dwork:roth:14}.

Using the Laplace tail bound and a union bound over all $N=k^{h-1}$
nodes ($k$ is the arity),
\begin{align}
\Pr\!\left[\max_{u\in V}|Z_u|>E\right]
\le
N\exp\!\left(-\frac{\varepsilon E}{h}\right).
\end{align}
Therefore, a sufficient simultaneous classification margin is
\begin{align}
E_{\mathrm{Lap}}
:=
\frac{h}{\varepsilon}\ln\frac{N}{\beta}
=
O\!\left(
\frac{h}{\varepsilon}\ln\frac{N}{\beta}
\right)=O\!\left(\frac{\ln(1/\beta)}{\varepsilon}h^2
\right),
\label{eq:baseline-laplace-margin}
\end{align}
with $h=O(\log N)$ substituted.

% Using the advanced composition~\cite{boosting-and-dp,pmlr-v37-kairouz15}, we can choose the noise distribution $\mathcal{P}$ to be $\Lap\!\left(\frac{b\sqrt{h}}{\varepsilon}\right)$, where $b$ is some factor depending on the privacy parameter $\delta$. With algebraic substitutions, the sufficient simultaneous classification margin is 
% \begin{align}
% E_{\mathrm{Lap}, \text{advanced}}
% =
% O\!\left(\frac{\ln(1/\beta)}{\varepsilon}h\sqrt{h}
% \right),
% \label{eq:baseline-laplace-margin}
% \end{align}

\paragraph{Gaussian baseline.}
For $0<\varepsilon<1$ and $0<\delta<1/2$, choose the noise distribution $\mathcal{P}$ to be $\mathcal{N}(0,\sigma^2)$, with
\begin{align}
\sigma^2
=
\frac{2h\ln(1.25/\delta)}{\varepsilon^2}.
\label{eq:baseline-gaussian-noise}
\end{align}
Adding or removing one individual changes at most $h$
coordinates of the full count vector, each by one.
Its $\ell_2$-sensitivity is therefore at most $\sqrt{h}$.
Consequently, the vector Gaussian
mechanism~\cite{dwork:roth:14} gives
$(\varepsilon,\delta)$-DP. The Gaussian tail bound and the union bound give
\[
\Pr\!\left[\max_{u\in V}|Z_u|>E\right]
\le
2N\exp\!\left(-\frac{E^2}{2\sigma^2}\right).
\]
Thus, a sufficient simultaneous classification margin is
\begin{align}
E_{\mathrm{Gauss}}
&:=
\sigma\sqrt{2\ln\frac{2N}{\beta}}\\
&=
\frac{2}{\varepsilon}
\sqrt{
h\ln\frac{1.25}{\delta}
 \ln\frac{2N}{\beta}
}\\
&=
O\!\left(
\frac{
\sqrt{h\ln(1/\delta)\ln(2N/\beta)}
}{\varepsilon}
\right)\\
&=O\!\left(
\frac{
\sqrt{\ln(1/\delta)\ln(1/\beta)}
}{\varepsilon}
h\right).
\label{eq:baseline-gaussian-margin}
\end{align}
More advanced mechanisms aim at sublinear dependencies on $h$, as reviewed in Section~\ref{sec:literature}.

\subsection{Additive Error Guarantee of Ghazi et al.}
Ghazi et al.~\cite{ghazi2023} originally uses a relative tolerance $\gamma\in[0,1)$ rather than the additive absolute tolerance $E$. It also includes a public mass bound $M\ge n$. Its required answers are \emph{Heavy} for $c_u(D)\ge(1+\gamma)\tau$ and \emph{Light} for $c_u(D)<(1-\gamma)\tau$, provided $n\le M$. We show how to translate their relative error guarantee to our absolute error guarantee.

Define
\begin{align}\label{eq:LambdaG}
\Lambda_G(h,\beta,\eps,\delta)=\max\left\{48\ln\frac{2h}{\beta},\;
 6\ln\left(1+\frac{e^{\eps/2}-1}{\delta}\right)\right\}.
\end{align}
Lemma 4.3 of \cite{ghazi2023} gives an $(\eps,\delta)$-DP algorithm with a relative tolerance $\gamma\in(0,1/2)$, whenever
\begin{align}\label{eq:ghazithreshold}
 \tau\ge\sqrt{\frac{2M\Lambda_G}{\gamma\eps}}.
\end{align}
We consider failure probability $0<\beta<1/2$, following~\cite{ghazi2023}. At a fixed $\tau$, rearranging Eq.~\eqref{eq:ghazithreshold} gives
\begin{align}
 E_G&:=\gamma\tau\\\
 &\ge\ \frac{2M\Lambda_G}{\eps\tau}\\
 &=\frac{2M}{\eps\tau}\max\left\{48\ln\frac{2h}{\beta},\;
 6\ln\left(1+\frac{e^{\eps/2}-1}{\delta}\right)\right\}.
\end{align}
Therefore, at a prescribed threshold $\tau$, ~\cite{ghazi2023} provides an algorithm that certifies an additive margin in $O(\log h)$.
% ; at a prescribed relative tolerance $\gamma$, the corresponding sufficient threshold is in $O(\sqrt{\log h})$. 
% We can also translate our utility definition (see Eq.~\eqref{eq:utility}) to theirs by considering $ 0<\gamma<1,$  and set $\tau\ge E/\gamma.$

\subsection{Other Related Works}\label{sec:related-work-others}
\vspace{0pt}
\noindent\textbf{PrivTree.} PrivTree \cite{zhang2016privtree} provides depth-independent noise calibration for private hierarchical decomposition, but its splitting rule includes a depth-dependent bias. In particular, for node $v$ at depth $d$, its noisy count is compared with a threshold biased by $O(d)$. It does not automatically achieve a bounded error margin and is not comparable with our method. 

% PrivTree has depth-independent noise, but tests a \emph{depth-biased} score, roughly $c(v)-\operatorname{depth}(v)\,d_0$ subject to a floor, not the same raw-count threshold at every depth \cite[Algorithm~2]{privtree2016}. Consequently a long constant-count path need not be kept until the common raw count falls below $\tau$. Its unbounded noise also does not supply our probability-one count-margin guarantee. This is a mismatch of utility specifications, not a claim that PrivTree is invalid or that no adaptation can work.

\vspace{0pt}
\noindent\textbf{Threshold-release reduction.} In Section 1.3 of~\cite{ghazi2023}, Ghazi et al. discuss a reduction from the tree counting problem to the learning thresholds problem~\cite{bnsv2015}. This would give an additive error
\begin{equation}\label{eq:thresholdrelease}
 O_{\eps,\delta}\!\left((\log^*|\mathcal X|)^{O(1)}(\log n)^{2.5}\right).
\end{equation}
We note that the error is dependent on the size of the input database, $n$. It is therefore not directly compared with our method: Our method has a size-independent 
% and deterministic 
error margin while the error of~\cite{bnsv2015} 
% is probabilistic and 
depends on the input size due to the $(\log n)^{2.5}$ factor. 
% Cohen et al.~\cite{optimal-dp-thresholds} further improve the error dependency on $h$ for the learning thresholds problem. Still, using the reduction from~\cite{ghazi2023} and~\cite{bnsv2015}, the resulting error guarantee for heavy hitter classification is dependent on the input database size. For the same reason, we exclude it from our comparison.

%% file: App-algo-details.tex
\section{Further Algorithmic Details for \sc{BetweenCut}}

\subsection{Label Propagation}
\label{sec:app-marking}

Algorithm~\ref{alg:marking-routines} implements the deterministic label propagation routines used by Algorithm~\ref{alg:tree}. The routines inspect only the public tree and the current labels and make no additional count queries.

\begin{algorithm}[t]
\caption{Marking routines}\label{alg:marking-routines}
\begin{algorithmic}[1]
\Procedure{MarkUp}{$u$}
    \While{$u \neq \mathrm{null}$ \textbf{and}
           $\mathsf{label}[u] \neq \Hi$}
        \State $\mathsf{label}[u] \gets \Hi$
        \State $u \gets \operatorname{parent}(u)$
    \EndWhile
\EndProcedure
%\Statex
\Procedure{MarkDown}{$u$}
    \If{$\mathsf{label}[u] = \Lo$}
        \State \Return
    \EndIf
    \State $\mathsf{label}[u] \gets \Lo$
    \For{each child $v$ of $u$, in public child order}
        \State \Call{MarkDown}{$v$}
    \EndFor
\EndProcedure
\end{algorithmic}
\end{algorithm}

\subsection{Algorithm Properties}\label{app:algorithm-properties}

\begin{lemma}[Consistency and termination]\label{lem:consistency-termination}
For every execution of Algorithm~\ref{alg:tree} on a finite rooted tree $T$, the
following statements hold.
\begin{enumerate}
    \item[(i)] \textbf{Structural consistency.}
    Immediately after initialization and after each completed iteration
    of the inner node loop (lines~5--15), the Heavy set is
    ancestor-closed, the Light set is descendant-closed, and the two
    sets are disjoint.

    \item[(ii)] \textbf{Permanent labels.}
During Algorithm~\ref{alg:tree}, every label assignment after initialization
changes an unknown label to Heavy or Light: $
? \longrightarrow \mathrm{H}, 
? \longrightarrow \mathrm{L}.$ Consequently, each node is labeled at most once, and its label thereafter remains unchanged.

    \item[(iii)] \textbf{Termination and completeness.}
    Algorithm~\ref{alg:tree} terminates, and every node is labeled Heavy or Light
    when it returns.
\end{enumerate}
\end{lemma}

\begin{lemma}[Positive-count query bound]
\label{lem:positive-count-queries}
Fix a database $D$ of size $n$. For any execution of Algorithm~\ref{alg:tree}, let $\mathcal Q$ be the set of nodes at which $\GBT$ (Algorithm~\ref{alg:guard}) is invoked, and define
\begin{align}
    \mathcal Q_+
    := \{u\in\mathcal Q : c_u(D)>0\}.
\end{align}
Then
\begin{align}
    |\mathcal Q_+|
    \le
    \sum_{u\in\mathcal Q} c_u(D)
    \le n\lceil\log_2(h+1)\rceil.
\end{align}
\end{lemma}

\subsubsection{Dense implementation cost}
\label{sec:app-dense-cost}
We analyze the dense implementation of Algorithm~\ref{alg:tree}.

\vspace{0pt}
\noindent\textbf{Computational conventions.}
We assume that accessing counts and labels, moving along one tree edge,
arithmetic, comparisons, and Laplace sampling each take constant time.
Each node identifier, count, or label occupies constant space.
Each record gives direct access to its leaf, and the public height $h$
is given. Auxiliary space excludes the input database and the given
public tree.

\begin{lemma}[Dense implementation cost]
\label{lem:spv4-dense-cost}
Let $n:=|D|$ and $N:=|V|$. For every execution of
Algorithm~\ref{alg:tree}, the following hold.
\begin{enumerate}
    \item[(i)] \textbf{Scans, queries, and assignments.}
    The inner node loop considers each node exactly once. There are at
    most $N$ calls to \textsc{GuardedBT} and exactly $N$ marking
    assignments (i.e., labeling a node to $\Hi$ or $\Lo$) after initialization.

    \item[(ii)] \textbf{Total cost.}
    Under the computational conventions above, the total marking work
    is $O(N)$. Including count computation, level-list construction,
    traversal, and output, the algorithm can be implemented in
    $O(n+N)$ time and $O(N)$ auxiliary space.
\end{enumerate}
\end{lemma}

\subsection{Proofs of Algorithmic Properties}

\begin{proof}[Proof of Lemma~\ref{lem:consistency-termination}]
\textbf{(i)--(ii) Consistency and permanent labels.}
We induct on the completed iterations of the inner node loop.
Initially, both labeled sets are empty, and skipping an already
labeled node changes nothing. Otherwise, let $u$ be the queried
node and $H,L$ the current Heavy and Light sets. Since $u$ is
unknown, the closure invariant implies
\[
    \operatorname{Anc}(u)\cap L=\varnothing,
    \qquad
    \operatorname{Desc}(u)\cap H=\varnothing.
\]
Thus upward marking encounters no Light node, and downward
marking encounters no Heavy node. In the Between case, the
upward call modifies only proper ancestors of $u$, disjoint
from $\operatorname{Desc}(u)$, so it cannot introduce a conflict
for the downward call. Both routines stop before assigning to
a node already carrying the intended label. Hence every
assignment changes an unknown label to Heavy or Light.

By the closure invariant, early stopping omits only ancestors
already Heavy or subtrees already Light. Consequently, High
adds $\operatorname{Anc}(u)$ to $H$, Low adds
$\operatorname{Desc}(u)$ to $L$, and Between adds
$\operatorname{Anc}(u)\setminus\{u\}$ to $H$ and
$\operatorname{Desc}(u)$ to $L$. These updates preserve closure
and disjointness, completing the induction.

\medskip
\noindent\textbf{(iii) Termination and completeness.}
The recursive intervals in \textsc{DepthOrder} strictly shrink
and partition the current interval except for its middle depth.
Thus the recursion terminates and lists each depth exactly once,
so every node is considered exactly once.
Each query terminates, and each marking call traverses a finite
ancestor chain or subtree; hence Algorithm~\ref{alg:tree}
terminates. When considered, a node is either already labeled
or receives a label from its query, regardless of the answer.
By~(ii), this label persists, so every node is labeled at return.
\end{proof}

\begin{proof}[Proof of Lemma~\ref{lem:path-query-bound}]
We induct on the number $s$ of levels in a recursive interval
of \textsc{DepthOrder}, proving that it contains at most
$\lceil\log_2(s+1)\rceil$ queried nodes of any root-to-leaf
path $P$ in any execution. The case $s=0$ is immediate.

For $s\ge1$, processing the middle level contributes at most
one query on $P$ and leaves every node at that level labeled.
By the closure and permanence properties of
Lemma~\ref{lem:consistency-termination}, all remaining queries
on $P$ within the interval lie in at most one subinterval,
containing at most $\lfloor s/2\rfloor$ levels.
The induction hypothesis therefore bounds the total by
\[
    1+\left\lceil\log_2\bigl(\lfloor s/2\rfloor+1\bigr)\right\rceil
    =\lceil\log_2(s+1)\rceil.
\]
Taking $s=h$ proves the claim.
\end{proof}

\begin{proof}[Proof of Lemma~\ref{lem:positive-count-queries}]
\textbf{Counting identity.}
List the records of $D$ as $z_1,\ldots,z_n$, including repetitions.
A record $z_r$ contributes to $c_u(D)$ exactly when $u\in P_{z_r}$.
Interchanging the two finite sums gives
\begin{align}
    \sum_{u\in\mathcal Q} c_u(D)
    &=
    \sum_{u\in\mathcal Q}
    \sum_{r=1}^{n} \mathbf{1}\{u\in P_{z_r}\}=
    \sum_{r=1}^{n}
    \sum_{u\in\mathcal Q} \mathbf{1}\{u\in P_{z_r}\}=
    \sum_{r=1}^{n} |\mathcal Q\cap P_{z_r}|.
\end{align}

\medskip
\noindent
\textbf{Query bound.}
Every node in $\mathcal Q_+$ has an integer count of at least one,
so
\begin{align}
    |\mathcal Q_+|
    \le \sum_{u\in\mathcal Q} c_u(D).
\end{align}
By the counting identity and Lemma~\ref{lem:path-query-bound},
\begin{align}
    \sum_{u\in\mathcal Q} c_u(D)
    =
    \sum_{r=1}^{n} |\mathcal Q\cap P_{z_r}|
    \le
    \sum_{r=1}^{n} \lceil\log_2(h+1)\rceil.
\end{align}
\end{proof}

\begin{proof}[Proof of Lemma~\ref{lem:one-middle-target}]Let $\mathcal B$ be the set of
queried nodes at which Algorithm~\ref{alg:guard} returns $\Bt$.  We prove the following two items:
\begin{enumerate}
\item[(i)] If a query at $u$ returns $\Bt$, then, after the corresponding inner-loop iteration, every node on every root-to-leaf path through $u$ is labeled. No subsequent query occurs on any such path.
\item[(ii)] Every root-to-leaf path $P$ satisfies $|\mathcal B\cap P|\le 1.$
\end{enumerate}
Proof for (i).
When a query at $u$ returns $\mathrm{B}$, the algorithm calls
\textsc{MarkUp} on $\operatorname{parent}(u)$ and
\textsc{MarkDown} on $u$. After these calls, every proper ancestor
of $u$ is Heavy, and $u$ and all its descendants are Light.
Every node on a root-to-leaf path through $u$ belongs to one of
these two sets. By Lemma~\ref{lem:consistency-termination}, these labels are permanent, and
Algorithm~5 skips labeled nodes. Hence no subsequent query
occurs on such a path.

\noindent Proof for (ii). Suppose two distinct nodes of $\mathcal B$ were comparable
under the ancestor relation. They would lie on a common root-to-leaf path. By~(i), the Between answer at whichever node
was queried first would prevent the other from being queried, leading to a contradiction. Thus $|\mathcal B\cap P|\le 1$ for every root-to-leaf path $P$.
\end{proof}

\begin{proof}[Proof of Lemma~\ref{lem:spv4-dense-cost}]
\textbf{(i) Scans, queries, and assignments.}
The depth order lists each depth exactly once, and each node appears
in exactly one level list. Only unknown nodes are queried, so there
are at most $N$ calls to \textsc{GuardedBT}. By
Lemma~\ref{lem:consistency-termination}, every marking assignment
labels a previously unknown node and every node is eventually labeled.
Hence there are exactly $N$ such assignments.

\medskip
\noindent\textbf{(ii) Total cost.}
Each successful \textsc{MarkUp} iteration labels a previously unknown
node, so there are at most $N$ such iterations in total. Each invocation
has at most one final stopping test. \textsc{MarkDown} scans a node's
child list only when that node first becomes Light, so each of the
$N-1$ tree edges is scanned at most once by this routine. Every recursive
downward call is charged to such an edge. Each query initiates at most
one upward and one downward call. Together with (i), these bounds give
$O(N)$ total marking work, without requiring bounded arity.

Initialize all counts to zero, tally the $n$ leaf records, and sum child
counts in a bottom-up traversal, taking $O(n+N)$ time. An ordered
depth-first traversal builds the level lists in $O(N)$ time, with each
list following the order induced by the public child order. The depth
order is generated in $O(h)$ time. Label initialization, node scans,
\textsc{GuardedBT} calls, and output take $O(N)$ additional time. Since
$h\le N$, the total time is $O(n+N)$. Counts, labels, level lists, and
traversal stacks use $O(N)$ auxiliary space.
\end{proof}

%% file: App-proofs.tex
\section{Privacy and Utility Proofs}\label{sec:app-proof}

\subsection{Proof for the Bespoke Analysis for TCT}\label{sec:app-proof-one-hit}

\begin{proof}[Proof of Lemma~\ref{lem:tct-one-hit}]
Fix neighboring inputs $D_0,D_1$ and fix the analyst's
data-independent randomness. We work with an analysis-only augmented
transcript.

At each call, use the decomposition
\begin{align}
  P_b=(1-\kappa)\bigl(pC+(1-p)N_b\bigr)+\kappa F_b
\end{align}
to generate the output as follows. First define a branch in \emph{exceptional}, \emph{common}, and \emph{residual}
\begin{align}
  J\in\{\mathsf{exc},\mathsf{com},\mathsf{res}\}
\end{align}
with probabilities
\begin{align}
  \kappa,\qquad (1-\kappa)p,\qquad
  (1-\kappa)(1-p),
\end{align}
respectively. Conditional on the branch, sample from $F_b$, $C$, or
$N_b$, respectively. The parameters of the decomposition may depend on the preceding
transcript, but for a fixed preceding history the branch probabilities
are the same under $D_0$ and $D_1$. 

\paragraph{Exceptional branches.}
Let
\begin{align}
  E_{\mathrm{exc}}
  :=
  \{\text{an exceptional branch is selected at some call}\}.
\end{align}
For $i\in\{1,\ldots,K\}$, let $E_i$ be the event that the $i$th call
occurs and selects the exceptional branch. Conditional on every
history preceding that call, its exceptional probability is at most
$\kappa_0$. Hence, for either $b\in\{0,1\}$,
\begin{align}
  \Pr_b[E_i]\le \kappa_0,
\end{align}
and therefore
\begin{equation}
  \Pr_b[E_{\mathrm{exc}}]
  \le \sum_{i=1}^{K}\Pr_b[E_i]
  \le K\kappa_0.
  \label{eq:tct-exception-bound}
\end{equation}

\paragraph{Residual overflow.}
Call a sample from $N_b$ a \emph{residual draw}. Define the overflow event:
\begin{align}
  E_{\mathrm{ov}}
  :=
  \{\text{at least $r+1$ residual draws occur before any exception}\}.
\end{align}
For $j\ge 0$, define $A_j$ to be the event that the first $j$
residual draws all occur before any exceptional branch and all $j$
outputs miss their respective targets. In particular,
$A_0$ is satisfied with certainty.

We claim that, for either input,
\begin{equation}
  \Pr_b[A_j]\le (1-q)^j.
  \label{eq:tct-Aj-bound}
\end{equation}
Suppose $A_{j-1}$ has occurred. The interaction may terminate
before another residual draw is reached. If a $j$th residual draw is
reached, however, its conditional law is $N_b$, and therefore its
conditional probability of missing the target is at most $1-q$.
Consequently,
\begin{align}
  \Pr_b[A_j\mid A_{j-1}]\le 1-q,
\end{align}
and induction gives~\eqref{eq:tct-Aj-bound}.

If $E_{\mathrm{ov}}$ occurs, then the first $r$ residual draws must all
miss their targets; otherwise a target hit would have stopped the
interaction before the $(r+1)$st residual draw. Thus
\begin{align}
  E_{\mathrm{ov}}\subseteq A_r,
\end{align}
and hence
\begin{equation}
  \Pr_b[E_{\mathrm{ov}}]\le (1-q)^r.
  \label{eq:tct-overflow-bound}
\end{equation}
When $r=K$, the event $E_{\mathrm{ov}}$ is impossible, because there
are at most $K$ calls in total.

\paragraph{The good event.}
Let
\begin{align}
  E_{>r}
  :=
  \{\text{more than $r$ residual draws occur during the interaction}\},
\end{align}
and define
\begin{align}
  G
  :=
  E_{\mathrm{exc}}^c \cap E_{>r}^c.
\end{align}
Thus $G$ is the event that no exceptional branch occurs and at most
$r$ residual draws occur.

By De Morgan's law,
\begin{equation}
  G^c
  =
  E_{\mathrm{exc}}\cup E_{>r}.
  \label{eq:tct-good-complement}
\end{equation}
Decompose the second event according to whether an exceptional branch
occurs:
\begin{align}
  E_{>r}
  =
  (E_{\mathrm{exc}}\cap E_{>r})
  \cup
  (E_{\mathrm{exc}}^c\cap E_{>r}).
\end{align}
Since
\begin{align}
  E_{\mathrm{exc}}\cap E_{>r}
  \subseteq E_{\mathrm{exc}},
\end{align}
we obtain
\begin{align}
  G^c
  =
  E_{\mathrm{exc}}\cup E_{>r}
  \subseteq
  E_{\mathrm{exc}}
  \cup
  (E_{\mathrm{exc}}^c\cap E_{>r}).
\end{align}

Recall the overflow event
\begin{align}
  E_{\mathrm{ov}}
  :=
  \{\text{the $(r+1)$st residual draw occurs before any exceptional
  branch}\}.
\end{align}
On the event $E_{\mathrm{exc}}^c\cap E_{>r}$, no exceptional branch
occurs anywhere in the interaction and at least $r+1$ residual draws
occur. Therefore the $(r+1)$st residual draw necessarily occurs before
any exceptional branch, and hence
\begin{align}
  E_{\mathrm{exc}}^c\cap E_{>r}
  \subseteq E_{\mathrm{ov}}.
\end{align}
Consequently,
\begin{equation}
  G^c
  \subseteq
  E_{\mathrm{exc}}\cup E_{\mathrm{ov}}.
  \label{eq:tct-good-bad-inclusion}
\end{equation}
Taking probabilities and applying the union bound gives, for either
$b\in\{0,1\}$,
\begin{align}
  \Pr_b[G^c]
  \le
  \Pr_b[E_{\mathrm{exc}}]
  +
  \Pr_b[E_{\mathrm{ov}}]
  \le
  K\kappa_0
  +
  \mathbf{1}_{\{r<K\}}(1-q)^r.\label{eq:tct-good-event-bound}
\end{align}

Now fix an augmented transcript $\omega\in G$. Since no exceptional
branch occurs, every factor in its probability comes either from a
common branch or from a residual branch. The branch-selection
probabilities are identical under the two inputs, and the common
distribution $C$ is also identical. Hence these factors cancel in the
likelihood ratio. For each residual draw,
\begin{align}
  N_0(y)\le e^\lambda N_1(y)
  \qquad\text{and}\qquad
  N_1(y)\le e^\lambda N_0(y).
\end{align}
Since $\omega\in G$ contains at most $r$ residual draws,
\begin{equation}
  \Pr_0[\omega]\le e^{r\lambda}\Pr_1[\omega],
  \qquad
  \Pr_1[\omega]\le e^{r\lambda}\Pr_0[\omega].
  \label{eq:tct-pointwise-good}
\end{equation}

Finally, let $S$ be any event in the released transcript space, and
let $\widetilde S$ be its inverse image in the augmented transcript
space. Using
Eq.~\eqref{eq:tct-good-event-bound} and
Eq.~\eqref{eq:tct-pointwise-good},
\begin{align*}
  \Pr_0[S]
  &\le
    \Pr_0[\widetilde S\cap G]+\Pr_0[G^c] \\
  &\le
    e^{r\lambda}\Pr_1[\widetilde S\cap G]
    +K\kappa_0
    +\mathbf{1}_{\{r<K\}}(1-q)^r \\
  &\le
    e^{r\lambda}\Pr_1[S]
    +K\kappa_0
    +\mathbf{1}_{\{r<K\}}(1-q)^r.
\end{align*}
Exchanging $D_0$ and $D_1$ gives the reverse inequality. 
\end{proof}

The important point for this Lemma that we will use subsequently is that $r$ is an analysis parameter that we are free to tune: increasing $r$ weakens the $\eps$ component but strengthens the $\delta$ component of the $(\eps, \delta)$-DP guarantee.  
We will use this freedom subsequently to choose a value of $r$ that achieves the target $(\eps,\delta)$ pair for the main theorem. 

\subsection{Proofs for Privacy}\label{sec:app-transcript-privacy}

\begin{proof}[Proof of Lemma~\ref{lem:BTtarget}]
Write $P_i:=P_{c_i}$ for $i\in\{0,1\}$, and exchange the indices
if necessary so that $c_0\le c_1$. Set
\begin{align}
  \xi:=e^\lambda>1,
  \qquad
  \gamma:=1-e^{-\lambda(b-a)}.
\end{align}
We will specify the probabilities of each output under $C,N_0,N_1$
and prove that
\begin{align}
  P_i=pC+(1-p)N_i,
  \qquad
  N_0\approx_\lambda N_1,
  \qquad
  N_i(\mathsf B)\ge\frac{\gamma}{1+\xi}.
\end{align}

\paragraph{Step 1: The original output probabilities.}
For $Z\sim\operatorname{Lap}(1/\lambda)$, define
\begin{align}
    \ell_i&:=P_i(\mathsf L)=\Pr[c_i+Z<a],\\
    m_i&:=P_i(\mathsf B)=\Pr[a\le c_i+Z\le b],\\
    u_i&:=P_i(\mathsf H)=\Pr[c_i+Z>b].
\end{align}
These are ordinary real numbers satisfying
$\ell_i+m_i+u_i=1$. In particular, $m_i>0$ because $a<b$ and the
Laplace density is positive everywhere. Since $c_0\le c_1$,
\begin{align}\label{eq:monotone-l-u}
  \ell_0\ge\ell_1,
  \qquad
  u_0\le u_1.
\end{align}
The densities $f_i(x)=(\lambda/2)e^{-\lambda|x-c_i|}$ satisfy $
  \frac{f_0(x)}{f_1(x)}
  =e^{\lambda(|x-c_1|-|x-c_0|)}
  \in[\xi^{-1},\xi],$
since $|c_1-c_0|\le1$. Integrating over the regions corresponding to
any set of output symbols gives
\begin{align}
  P_0(S)\le\xi P_1(S),
  \qquad
  P_1(S)\le\xi P_0(S),
  \qquad S\subseteq\{\mathsf L,\mathsf B,\mathsf H\}.
  \label{eq:bt-proof-original-private}
\end{align}

\paragraph{Step 2: Choose how much probability is shared.}
We assign the same amount $\alpha_{\mathsf L}$ of the Low probability
under each input to the \emph{common component}, and similarly assign
$\alpha_{\mathsf H}$ of the High probability. We assign \emph{no Middle
probability to the common component}.

The amounts must be chosen so that the remaining probabilities still
satisfy the factor-$\xi$ privacy bounds. For the Low and High outputs,
respectively, choose them to satisfy
\begin{align}
  \ell_0-\alpha_{\mathsf L}
    =\xi(\ell_1-\alpha_{\mathsf L}),
  \qquad
  u_1-\alpha_{\mathsf H}
    =\xi(u_0-\alpha_{\mathsf H}).
\end{align}
Solving these scalar equations gives
\begin{equation}
  \alpha_{\mathsf L}
    :=\frac{\xi\ell_1-\ell_0}{\xi-1},
  \qquad
  \alpha_{\mathsf H}
    :=\frac{\xi u_0-u_1}{\xi-1}.
  \label{eq:bt-proof-common-masses}
\end{equation}
These choices are valid fractions of probability mass to assign to the common component. 
Indeed, Eq.~\eqref{eq:bt-proof-original-private} gives $\ell_0\le\xi\ell_1$ and $u_1\le\xi u_0$, so both quantities are nonnegative. 
Also, by Eq.~\eqref{eq:monotone-l-u}
\begin{align}
    \alpha_{\mathsf L}
      &=\ell_1-\frac{\ell_0-\ell_1}{\xi-1}
        \le\ell_1\le\ell_i,\\
    \alpha_{\mathsf H}
      &=u_0-\frac{u_1-u_0}{\xi-1}
        \le u_0\le u_i.
\end{align}
Thus neither amount exceeds the corresponding output probability under either input.

Set
\begin{align}
  p:=\alpha_{\mathsf L}+\alpha_{\mathsf H},
  \qquad
  \rho:=1-p.
\end{align}
For each $i$, we have $0\le p\le\ell_i+u_i=1-m_i<1$, so
$\rho\ge m_i>0$.

\paragraph{Step 3: Define the distributions and verify the mixture.}
If $p>0$, define $C$ by its three output probabilities:
\begin{align}
  C(\mathsf L):=\frac{\alpha_{\mathsf L}}p,
  \qquad C(\mathsf B):=0,
  \qquad C(\mathsf H):=\frac{\alpha_{\mathsf H}}p.
\end{align}
If $p=0$, define $C(\mathsf L)=1$ and
$C(\mathsf B)=C(\mathsf H)=0$; this choice is immaterial because
$C$ then has mixture weight zero. In either case, $C$ is a probability
distribution and
\begin{align}
  pC(\mathsf L)=\alpha_{\mathsf L},
  \qquad pC(\mathsf B)=0,
  \qquad pC(\mathsf H)=\alpha_{\mathsf H}.
\end{align}

For each $i\in\{0,1\}$, define $N_i$ by
\begin{equation}
  \begin{aligned}
    N_i(\mathsf L)&:=\frac{\ell_i-\alpha_{\mathsf L}}\rho,\\
    N_i(\mathsf B)&:=\frac{m_i}\rho,\\
    N_i(\mathsf H)&:=\frac{u_i-\alpha_{\mathsf H}}\rho.
  \end{aligned}
  \label{eq:bt-proof-residual-probabilities}
\end{equation}
All three quantities are nonnegative by Step~2, and their sum is
\begin{align}
  \frac{\ell_i-\alpha_{\mathsf L}+m_i+u_i-\alpha_{\mathsf H}}\rho
  =\frac{1-p}\rho=1.
\end{align}
Hence $N_i$ is a probability distribution.

Now consider sampling from $C$ with probability $p$, and from $N_i$
with probability $\rho=1-p$. For each output, the resulting probability
is
\begin{align}
    pC(\mathsf L)+\rho N_i(\mathsf L)
      &=\alpha_{\mathsf L}+(\ell_i-\alpha_{\mathsf L})=\ell_i,\\
    pC(\mathsf B)+\rho N_i(\mathsf B)
      &=0+m_i=m_i,\\
    pC(\mathsf H)+\rho N_i(\mathsf H)
      &=\alpha_{\mathsf H}+(u_i-\alpha_{\mathsf H})=u_i.
\end{align}
Thus this sampling procedure reproduces $P_i$ exactly. Equivalently,
\begin{equation}
  P_i=pC+(1-p)N_i.
  \label{eq:bt-proof-mixture}
\end{equation}

\paragraph{Step 4: Verify privacy of the residual pair.}
The two equations used to choose the common amounts in Step~2, together
with the common normalizing factor $\rho$, give
\begin{align}
  N_0(\mathsf L)=\xi N_1(\mathsf L),
  \qquad
  N_1(\mathsf H)=\xi N_0(\mathsf H).
\end{align}
These identities also hold when the corresponding residual
probabilities are zero. At the Middle output,
\begin{align}
  N_0(\mathsf B)=\frac{m_0}\rho
    \le\xi\frac{m_1}\rho=\xi N_1(\mathsf B),
\end{align}
and the reverse inequality also holds  from Eq.~\eqref{eq:bt-proof-original-private}.
Consequently, $N_0\approx_\lambda N_1$.

\paragraph{Step 5: Bound the residual Middle probability.}
By construction, $N_i(\mathsf B)=m_i/\rho$. Thus it remains to show
that the remaining total probability $\rho$ is not too large relative
to $m_i$.

First, applying Eq.~\eqref{eq:bt-proof-original-private} to the non-Low
and non-High events gives
\begin{align}
  1-\ell_1\le\xi(1-\ell_0),
  \qquad
  1-u_0\le\xi(1-u_1).
\end{align}
Using the first inequality and $\alpha_{\mathsf H}\ge0$, we obtain
\begin{align}
  \rho:=1-(\alpha_{\mathsf L}+\alpha_{\mathsf H})
  &\le1-\alpha_{\mathsf L}
   =\frac{\xi(1-\ell_1)-(1-\ell_0)}{\xi-1}
   \le\frac{(\xi^2-1)(1-\ell_0)}{\xi-1}
   =(1+\xi)(1-\ell_0)
   \le(1+\xi)(1-\ell_i).
  \label{eq:bt-proof-residual-low}
\end{align}
Similarly, using the second inequality and $\alpha_{\mathsf L}\ge0$,
\begin{align}
  \rho
  &\le1-\alpha_{\mathsf H}
   =\frac{\xi(1-u_0)-(1-u_1)}{\xi-1}
   \le\frac{(\xi^2-1)(1-u_1)}{\xi-1}
   =(1+\xi)(1-u_1)
   \le(1+\xi)(1-u_i).
  \label{eq:bt-proof-residual-high}
\end{align}

Fix $i\in\{0,1\}$. There are three cases, depending only on the
location of $c_i$.

\emph{If $c_i\le a$,} the Laplace upper-tail formula gives
\begin{align}
  1-\ell_i=\frac12e^{-\lambda(a-c_i)},
  \qquad
  m_i=\frac12e^{-\lambda(a-c_i)}
             \bigl(1-e^{-\lambda(b-a)}\bigr)
      =\gamma(1-\ell_i).
\end{align}
Therefore, by Eq.~\eqref{eq:bt-proof-residual-low},
\begin{align}
  N_i(\mathsf B)
  =\frac{m_i}\rho
  \ge\frac{\gamma(1-\ell_i)}{(1+\xi)(1-\ell_i)}
  =\frac{\gamma}{1+\xi}.
\end{align}

\emph{If $c_i\ge b$,} the Laplace lower-tail formula gives
\begin{align}
  1-u_i=\frac12e^{-\lambda(c_i-b)},
  \qquad
  m_i=\frac12e^{-\lambda(c_i-b)}
             \bigl(1-e^{-\lambda(b-a)}\bigr)
      =\gamma(1-u_i).
\end{align}
Eq.~\eqref{eq:bt-proof-residual-high} yields
\begin{align}
  N_i(\mathsf B)
  \ge\frac{\gamma(1-u_i)}{(1+\xi)(1-u_i)}
  =\frac{\gamma}{1+\xi}.
\end{align}

\emph{If $a<c_i<b$,} set
\begin{align}
  x:=e^{-\lambda(c_i-a)},
  \qquad
  y:=e^{-\lambda(b-c_i)}.
\end{align}
Then $x,y\in(0,1)$, $xy=e^{-\lambda(b-a)}=1-\gamma$, and
\begin{align}
  m_i=1-\frac{x+y}{2}.
\end{align}
Since $(1-x)(1-y)\ge0$, we have
$x+y\le1+xy=2-\gamma$, and hence $m_i\ge\gamma/2$.
Because $\rho\le1$ and $1+\xi>2$,
\begin{align}
  N_i(\mathsf B)=\frac{m_i}\rho
    \ge m_i\ge\frac{\gamma}{2}
    \ge\frac{\gamma}{1+\xi}.
\end{align}

The bound holds for each $i$, regardless of the location of the other
count. In particular, it covers pairs lying on different sides of a
threshold. Finally, since $b-a=2/\lambda$,
\begin{align}
  \frac{\gamma}{1+\xi}
  =\frac{1-e^{-\lambda(b-a)}}{1+e^\lambda}
  =\frac{1-e^{-2}}{1+e^\lambda}
  =\qstar.
\end{align}
Together with Steps~3--4, this proves the required pure target
decomposition.
\end{proof}

\begin{proof}[Proof of Lemma~\ref{lem:guards}]
Write $P_i:=P_{c_i}$ for the unguarded output laws from
Lemma~\ref{lem:BTtarget}, and $Q_i:=Q_{c_i}$ for the guarded output
laws. Let $\Delta_s$ denote the probability distribution that returns
$s\in\mathcal Y$ with probability one. Whenever $\kappa=0$, choose
the unused distributions $F_0=F_1=\Delta_{\mathsf B}$.

Exchange the indices if necessary so that $c_0\le c_1$. If the
counts are distinct, $c_1=c_0+1$.

\paragraph{Identical guarded laws.}
If $c_0=c_1$, or both counts lie in the same deterministic region,
then $Q_0=Q_1$. Choose
\begin{align}
  \kappa=0,\qquad p=1,\qquad C=Q_0=Q_1,
  \qquad N_0=N_1=\Delta_{\mathsf B}.
\end{align}
These choices give the required decomposition, with identical
residual laws and residual target probability one.

\paragraph{Both counts are in the randomized region.}
Suppose $a-R<c_0\le c_1<b+R$. Then $Q_i=P_i$ for both inputs.
By Lemma~\ref{lem:BTtarget}, there exist $p\in[0,1]$ and
probability distributions $C,N_0,N_1$ such that
\begin{align}
  Q_i=P_i=pC+(1-p)N_i,\qquad i\in\{0,1\},\\
  N_0\approx_\lambda N_1,\qquad
  \min_i N_i(\mathsf B)\ge\qstar.
\end{align}
The hypotheses of that lemma hold because $|c_1-c_0|\le1$ and
$b-a=2/\lambda$. Taking $\kappa=0$ therefore gives the desired
guarded decomposition without changing any of these witnesses.

\paragraph{One guard is crossed.}
The counts cannot lie in opposite deterministic regions: that would
require $c_1-c_0\ge b-a+2R>1$. Thus the only remaining cases have
one randomized input and one deterministic input. Define
\begin{align}
  (j,s,d):=
  \begin{cases}
    (1,\mathsf L,a-c_1),
      & c_0\le a-R<c_1,\\
    (0,\mathsf H,c_0-b),
      & c_0<b+R\le c_1.
  \end{cases}
\end{align}
Here $j$ indexes the randomized input, and $s$ is the deterministic
output on the other input. Hence
\begin{align}
  Q_j=P_j,\qquad Q_{1-j}=\Delta_s.
\end{align}
In the lower-guard case, $c_1=c_0+1\le a-R+1$; in the upper-guard
case, $c_0=c_1-1\ge b+R-1$. Consequently, in either case,
$d\ge R-1\ge0$.

Define $\kappa$ to be the probability that the randomized input
does not return $s$. The Laplace tail formulas used in the proof of
Lemma~\ref{lem:BTtarget} give
\begin{align}
  \kappa
  :=P_j(\mathcal Y\setminus\{s\})
  =\frac12e^{-\lambda d}
  \le\frac12e^{-\lambda(R-1)}
  =\delta_0.
\end{align}
In particular, $0<\kappa\le\delta_0\le1/2$. Choose
\begin{align}
  p=1,\qquad C=\Delta_s,\qquad
  N_0=N_1=\Delta_{\mathsf B},\qquad
  F_{1-j}=\Delta_s,
\end{align}
and define the remaining exceptional distribution by its output
probabilities:
\begin{align}
  F_j(y):=
  \begin{cases}
    0, & y=s,\\
    P_j(y)/\kappa, & y\ne s.
  \end{cases}
\end{align}
This is a probability distribution because
$\sum_{y\ne s}P_j(y)=\kappa$. Moreover, $P_j(s)=1-\kappa$, so,
for every $y\in\mathcal Y$,
\begin{align}
  Q_j(y)
    &=(1-\kappa)C(y)+\kappa F_j(y),\\
  Q_{1-j}(y)
    &=C(y)
     =(1-\kappa)C(y)+\kappa F_{1-j}(y).
\end{align}
Since $p=1$, these identities give exactly
\begin{align}
  Q_i=(1-\kappa)\bigl(pC+(1-p)N_i\bigr)+\kappa F_i,
  \qquad i\in\{0,1\}.
\end{align}
The residual distributions have zero mixture weight, but still
satisfy the conditions of the target definition:
\begin{align}
  N_0=N_1=\Delta_{\mathsf B},\qquad
  N_0\approx_\lambda N_1,\qquad
  N_i(\mathsf B)=1\ge\qstar.
\end{align}

These cases exhaust all neighboring count pairs. In every case,
the constructed witnesses satisfy the required residual privacy
and target bounds, with exceptional mass $\kappa\le\delta_0$.
\end{proof}

\paragraph{Transcript privacy.}
For a database $D$, let
\begin{align}
  \operatorname{Tr}(D)
  :=((u_1,s_1),\ldots,(u_m,s_m))
  \label{eq:query-transcript}
\end{align}
denote the complete query transcript of Algorithm~\ref{alg:tree},
where $u_j$ is the node queried at call $j$, $s_j\in\mathcal Y$
is its returned symbol, and $m$ is the total number of calls. Set
$L:=\lceil\log_2 (h+1)\rceil$. Substituting parameters from Lemma~\ref{lem:guards} to Lemma~\ref{lem:one-middle-target}, we obtain the following privacy guarantees. 

\begin{lemma}[Transcript privacy]
\label{lem:transcript-privacy}
Let $T$ be a finite public rooted tree of height $h\ge1$. Fix public parameters $\tau\ge 0$, $0<\lambda\le 1$,
and $R\ge 1$.
Run Algorithm~\ref{alg:tree} with thresholds
$a:=\tau-1/\lambda$ and $b:=\tau+1/\lambda$, using fresh
independent noise in every randomized invocation of $\GBT$. Let $\qstar$ and $\delta_0$ be as in
Eqs.~\eqref{eq:qexact} and~\eqref{eq:localdelta}, respectively.
For every neighboring pair $D_0\sim D_1$ and every integer
$1\le r\le L$, the transcripts
$\operatorname{Tr}(D_0)$ and $\operatorname{Tr}(D_1)$ obtained from Algorithm~\ref{alg:tree} are $\left(r\lambda,\delta_r\right)\text{-indistinguishable}$ for
\begin{align}
\delta_r:=L\delta_0+
    \mathbf{1}_{\{r<L\}}(1-\qstar)^r.
  \label{eq:transcript-privacy}
\end{align}
\end{lemma}

\begin{proof}[Proof of Lemma~\ref{lem:transcript-privacy}]
Fix neighboring databases and orient them as
$D_1=D_0\uplus\{z\}$, where $z$ is a leaf record. Let $P_z$ be
its root-to-leaf path. The original node counts satisfy
\begin{align}
  c_{u}(D_1)-c_{u}(D_0)
  =\mathbf{1}_{\{u\in P_z\}}.
  \label{eq:transcript-count-difference}
\end{align}
Call a query \emph{affected} if its node belongs to $P_z$.
For a fixed preceding query-response transcript, the next queried
node and any stopping decision are determined by the public
traversal order and the labels obtained from earlier responses.

\paragraph{Fix the off-path randomness.}
Recall  from
Lemma~\ref{lem:guards} that $Q_c$ is the guarded output law of $\GBT$ (Algorithm~\ref{alg:guard}) on input count $c$. For each $u\notin P_z$, independently
pre-sample an output
\begin{align}
  W_u\sim Q_{c_{u}(D_0)}=Q_{c_{u}(D_1)},
  \label{eq:off-path-samples}
\end{align}
independently also of all randomness used at affected queries.
The equality follows from Eq.~\eqref{eq:transcript-count-difference}.
Write $W=(W_u)_{u\notin P_z}$ and let $\mu$ denote its law.
For this fixed neighboring pair, $\mu$ is the same under both inputs. Algorithm~\ref{alg:tree} invokes $\GBT$ (Algorithm~\ref{alg:guard}) at each node at most once. Consequently,
using $W_u$ as the answer whenever an off-path node $u$ is queried,
and using fresh randomness at affected queries, reproduces the
original transcript distribution exactly. 

\paragraph{Apply Lemma~\ref{lem:tct-one-hit} to the affected queries.}
Fix a realization $w$ with $\mu(w)>0$. Consider a simulator that
knows $D_0,D_1,w$, but whose private input is only the bit
$i\in\{0,1\}$. It follows Algorithm~\ref{alg:tree}, answers every off-path query
at $u$ by $w_u$, and obtains an affected response at $u\in P_z$
by invoking the two-input mechanism
\begin{align}
  \mathcal M_u(i)\sim Q_{c_{u}(D_i)}
  \qquad\text{with target }\{\mathsf B\}.
\end{align}
Let $A_w(i)$ denote the query-response transcript of these affected
invocations only. Since $w$ is fixed, the simulator's next affected
request and its decision that no affected requests remain are
functions of the preceding affected transcript, the fixed quantities
$D_0,D_1,w$, and the public tree and traversal order. The affected
responses use fresh randomness, so their conditional laws are exactly
the displayed laws of $\mathcal M_u$.

For every affected request, the two counts differ by one.
Lemma~\ref{lem:guards} therefore gives the required target
decomposition with local parameter $\lambda$, target parameter
$\qstar$, and exceptional mass at most $\delta_0$. Then we can apply Lemma~\ref{lem:tct-one-hit} on the two-element input space
$\{0,1\}$, with $0$ and $1$ designated neighbors, and with
\begin{align}
  K=\lceil \log_2(h+1) \rceil,\qquad q=\qstar,\qquad \kappa_0=\delta_0.
\end{align}
It follows that $A_w(0)$ and $A_w(1)$ are
$(r\lambda,\delta_r)$-indistinguishable, uniformly over $w$.

\paragraph{Recover the complete transcript.}
For fixed $w$, the full transcript can be reconstructed by replaying
the public traversal, using $w_u$ at off-path queries and the
recorded responses in $A_w(i)$ at affected queries. This also
reconstructs any off-path queries after the last affected query.
Hence there is a deterministic map $\Phi_w$, the same for both
values of $i$, such that $\Phi_w(A_w(i))$ has the full transcript
law in the simulation with off-path answers $w$.

For any event $S$ in the complete transcript space:
% , post-processing for each fixed $w$ and averaging with the common weights $\mu(w)$ give
\begin{align}
  \Pr[\operatorname{Tr}(D_0)\in S]
  &=\sum_w\mu(w)
      \Pr[\Phi_w(A_w(0))\in S]\notag\\
  &\le\sum_w\mu(w)\left(
      e^{r\lambda}\Pr[\Phi_w(A_w(1))\in S]+\delta_r
      \right)\notag\\
  &=e^{r\lambda}\Pr[\operatorname{Tr}(D_1)\in S]+\delta_r.
\end{align}
Exchanging the inputs completes the proof for
Eq.~\eqref{eq:transcript-privacy}.
\end{proof}

The released heavy set is a deterministic function
of $\operatorname{Tr}(D)$ and the public tree and traversal order. Post-processing preserves DP, so the guarantee in Lemma~\ref{lem:transcript-privacy} also holds for Algorithm~\ref{alg:tree}, which was presented as Theorem~\ref{thm:guarded-tree-privacy}. 

\begin{proof}[Proof of Theorem~\ref{thm:calibrated-privacy}]
Since $h\ge 1$, we have $L:=\lceil\log_2(h+1)\rceil\ge 1$.
Also, $r_0\ge 1$, so $1\le r\le L$ and $0<\lambda=\frac{\varepsilon}{r}\le 1.$
Moreover, $L/\delta>1$ implies $R>1$, and $w=1/\lambda$. Theorem~\ref{thm:guarded-tree-privacy} applies.

The multiplicative privacy parameter is
$
    r\lambda=\varepsilon.
$
The guard term satisfies
\begin{align}
\label{eq:calibrated-guard-budget}
    \frac{L}{2}e^{-\lambda(R-1)}
    = \frac{L}{2}e^{-\ln(L/\delta)}
    = \frac{\delta}{2}.
\end{align}

If $r=L$, term $\mathbf{1}_{\{r<L\}}(1-\qstar)^r$ vanishes, proving
$(\varepsilon,\delta/2)$-DP.

Otherwise, $r=r_0<L$. Applying Theorem~\ref{thm:guarded-tree-privacy},
we obtain
\begin{align}
\label{eq:calibrated-target-budget}
    (1-q_\lambda)^r
    &\le \left(\frac45\right)^r = e^{-r\ln(5/4)} \le e^{-\ln(2/\delta)}
     = \frac{\delta}{2}.
\end{align}
Combining Eq.~\eqref{eq:calibrated-guard-budget} and
Eq.~\eqref{eq:calibrated-target-budget}, the total additive privacy
parameter is at most $\delta$.
This proves $(\varepsilon,\delta)$-DP.
\end{proof}

\subsection{Proofs for Utility}\label{sec:app-proof-utility}

\begin{theorem}[Exact deterministic margin]\label{thm:detutility}
% For every $\lambda>0$, $R\ge1$, and parameters Eq.~\eqref{eq:primitiveparams}, o
On any input database $D$, \cref{alg:tree} satisfies the following with probability $1$.
\begin{align}
 u\in\Hhat&\Longrightarrow c_u(D)>a-R=\tau-(w+R),\label{eq:heavycertificate}\\
 u\notin\Hhat&\Longrightarrow c_u(D)<b+R=\tau+(w+R).\label{eq:lightcertificate}
\end{align}
Consequently $E=w+R$ is a simultaneous deterministic margin as required in Definition~\ref{def:utility}.
\end{theorem}
Combining Theorem~\ref{thm:calibrated-privacy} and Theorem~\ref{thm:detutility} gives exactly the utility guarantees in Theorem~\ref{thm:main}. 

\begin{proof}[Proof of Theorem~\ref{thm:detutility}]
A query can return $\Hi$ or $\Bt$ only if its count is strictly greater than $a-R$, because counts at or below that lower guard return $\Lo$ deterministically. Every Heavy label has a queried witness: either that node or a descendant returned $\Hi$ or $\Bt$. The witness count is greater than $a-R$. Hence, its ancestor's count is greater than $a-R$. This proves Eq.~\eqref{eq:heavycertificate}. 

The proof for Eq.~\eqref{eq:lightcertificate} is similar. A query can return $\Lo$ or $\Bt$ only if its count is strictly less than $b+R$, because counts at or above that upper guard return $\Hi$ deterministically. Every Light node is the queried witness itself or a descendant of a query returning $\Lo$ or $\Bt$. Its count is at most the witness count, proving Eq.~\eqref{eq:lightcertificate}.

By termination every node receives one of these two labels, according to Lemma~\ref{lem:consistency-termination}. 
\end{proof}

%% file: App-sparse.tex
\section{Sparse Implementation}
\label{sec:spv4-sparse}
The dense implementation of Algorithm~\ref{alg:tree} stores and scans
all nodes of the public tree. The sparse implementation instead stores only the union of the root-to-leaf paths of the input records, removing branches that contain no records (namely, nodes with zero weight). When the threshold is at least the deterministic error margin, every zero-count query returns Low and affects only zero-count descendants. Such nodes can therefore be
omitted from invocations of $\GBT$, with the omitted nodes treated as Light. We apply the same
query and marking rules to the retained nodes, while classification
still concerns the entire original public tree. This preserves the
Heavy-set distribution and gives $O(nh)$ time and auxiliary space for an input
database of size $n$.

Throughout this section, write $n:=|D|$ and assume 
\[
    \tau\ge E:=w+R,\qquad w=1/\lambda,
    \qquad 0<\lambda\le1,\qquad R\ge1,
\]
with the original thresholds $a=\tau-w$ and $b=\tau+w$.

\paragraph{Stored subtree.}
Define
\[
    V_+:=\{u\in V:c_D(u)\ge1\}
        =\bigcup_{z\in D}P_z,
    \qquad T_+:=T[V_+],\qquad m:=|V_+|,
\]
where $P_z$ is the root-to-leaf path of record $z$. Thus $T_+$ contains
exactly the positive-count nodes and the original edges between them.
Every ancestor of a positive-count node also has positive count. $T_+$ is a rooted subtree of $T$ with the same root. It retains the original depths and public child order; no paths are contracted. Since each record
path contains at most $h$ nodes,
\begin{align}
    m\le\sum_{z\in D}|P_z|\le nh,
\end{align}
where repeated records are counted separately.

\paragraph{Implementation.}
If $D=\varnothing$, return $\varnothing$ immediately. Otherwise, build
$T_+$ by inserting each record path and increasing the count of every
node on that path by one, starting newly stored counts at zero. The
resulting counts are exactly the original counts $c_u(D)$.

Initialize all stored labels to unknown and apply the query and marking
steps of Algorithm~\ref{alg:tree} to the stored nodes. Keep the original
public parameters, height $h$, and depth order, even if $T_+$ is shorter.
At each depth, consider only stored nodes, in the restriction of the
original public order, and skip those already labeled. Use the same
marking routines: \textsc{MarkUp} follows the original ancestor chain,
all of which is stored, and \textsc{MarkDown} visits only stored children.
Their stopping rules are unchanged. Omitted nodes are implicitly Light. Return only the set of stored nodes labeled Heavy.

\begin{proposition}[Dense--sparse equivalence]
\label{prop:spv4-equivalence}
Under $\tau\ge E$, for every database $D$, the sparse implementation
and Algorithm~\ref{alg:tree} admit a coupling in which the sparse
query--response sequence is the dense sequence with all zero-count
queries deleted, and the final labels agree on all of $V$, interpreting
omitted nodes as Light. Consequently, they have the same Heavy-set
distribution and, with the same parameter calibration, the same
$(\varepsilon,\delta)$-DP guarantee and simultaneous deterministic
classification margin $E$.
\end{proposition}

\begin{proof}
For fixed public information, the dense traversal is determined by its responses: starting from unknown labels, each response fixes the marking
updates, and the public scan order fixes the next query or termination.

Since $a-R=\tau-E\ge0$, zero-count queries return Low and mark only zero-count descendants. Upward marking starts from a High or Between response at a positive-count node and stays within $V_+$. Thus zero-count queries do not affect stored labels, and omitted nodes never become Heavy.

Fix $D$ and couple the executions using independent
$Z_u\sim\operatorname{Lap}(1/\lambda)$ for $u\in V_+$, with both using $Z_u$ whenever the randomized branch at $u$ is reached. Since each node is queried at most once, this preserves the required fresh-noise distributions.

We can then make an inductive argument over the stored nodes in their common traversal order. Their labels initially agree, and intervening dense zero-count queries cannot change them. Both executions therefore skip the next stored node or query it with the same count and noise, obtaining the same response. Upward marking follows the same ancestor chain. Downward marking has the same effect on stored nodes, because omitted children have no stored descendants and the stopping tests agree on stored nodes. Hence the stored labels continue to agree, proving the query correspondence.

At termination, all omitted nodes are Light in both executions and all stored labels agree. Equality of the Heavy-set distributions transfers the established privacy and utility guarantees.
\end{proof}

% Sparse query selection also uses the private set $V_+$. The privacy
% guarantee therefore concerns only the released Heavy set, not $V_+$ or
% the sparse query log, which remain internal.

\paragraph{Running time and space.}
Use the computational conventions in Appendix~\ref{sec:app-dense-cost}, and assume bounded arity and $O(h)$-time access to each record path without enumerating the public tree. 
Inserting all paths builds $T_+$ and its counts in $O(nh)$ time. 
An ordered depth-first traversal of $T_+$ groups the stored nodes into one list per depth, preserving their inherited public order within each list. 
Building these lists and generating the original depth order by emitting each depth once together take $O(m+h)$ time and space.

Each stored node is considered once and queried at most once. 
The marking argument of Lemma~\ref{lem:spv4-dense-cost} gives $O(m)$ total marking work. 
Including construction, traversal, and output, the total time is $O(nh+m+h)$ and the auxiliary space is $O(m+h)$. Since $m\le nh$, both bounds are $O(nh)$. 
For $n=0$, the algorithm returns immediately in $O(1)$ time and auxiliary space.

%% file: App-extension-tree-counting.tex
\section{From Heavy-Node Classification to Numerical Tree Counting}
\label{sec:app-extension}

\subsection{Tree Counting through Classification}

Let $T$ be a finite public rooted tree with node set $V$ and height $h\ge3$. The private database $D$ consists of leaf records, and $c_u(D)$ is the count of node $u$, as in the main paper. 
The numerical tree counting problem~\cite{ghazi2023}
asks for an estimate $\widetilde c_u$ of every count. 

\paragraph{The $\operatorname{mRMSE}_{\alpha}$ error criterion.}
Fix a multiplicative tolerance $0<\alpha<1$. Following~\cite[Definition~1.4]{ghazi2023}, we allow an error of $\alpha c_u(D)$ and measure only the amount by which the absolute error exceeds this allowance. 
For an estimator $\widetilde c=\mathcal A(D)$, define the residual error at node $u$ by
\begin{align}
    r_u(D):=
    \bigl(|\widetilde c_u-c_u(D)|-\alpha c_u(D)\bigr)_+,
    \qquad (x)_+:=\max\{x,0\}.
\end{align}
For example, when $\alpha=0.1$ and $c_u(D)=1000$, an absolute error of $130$ gives residual error $30$. 
The multiplicative maximum root-mean-square error is
\begin{align}
    \operatorname{mRMSE}_{\alpha}(\mathcal A;T)
    :=\sup_D\max_{u\in V}
      \sqrt{\mathbb E\!\left[r_u(D)^2\right]},
    \label{eq:app-counting-mrmse}
\end{align}
where the expectation is over the estimator's randomness and the supremum ranges over all input databases. 
Thus $\operatorname{mRMSE}_{\alpha}(\mathcal A;T)\le B$ means that $\mathbb E[r_u(D)^2]\le B^2$ for every database and node.
The tolerance $\alpha$ is fixed in advance; it is the root-mean-square residual error that we bound as a function of $h$. 

\paragraph{The approach of Ghazi et al.}
Ghazi et al.~\cite[Section~4.2]{ghazi2023} use classification at geometrically spaced thresholds to obtain numerical estimates.
Their reduction processes larger thresholds first, assigns estimates to nodes identified as Heavy, and continues on a forest of lighter, disjoint subtrees at smaller thresholds. 
Classifications place estimated counts between nearby thresholds, giving multiplicative accuracy;
counts below the smallest scale are handled by an additive allowance. Combining this reduction with their SVT-based classifier gives an estimator
$\mathcal A_{\mathrm G}$ satisfying~\cite[Theorem~4.1]{ghazi2023} when $\alpha,\varepsilon,\delta$ are fixed
\begin{align}
    \operatorname{mRMSE}_{\alpha}(\mathcal A_{\mathrm G};T)
    =O_{\alpha,\varepsilon,\delta}(\log h).
\end{align}

\paragraph{Using \textsc{BetweenCut}.}
Our classifier (i.e., Algorithm~\ref{alg:tree}) has an additive classification margin $O_{\varepsilon,\delta}(\log\log h)$ that holds simultaneously for all nodes with probability one and is independent of the database size.
To address the tree counting problem, we can use the same geometric-threshold idea, but run increasing thresholds on the original tree. A node receives an estimate when it first becomes
Light. Unless this happens in the first round, its preceding Heavy answer gives a lower bound on its count, and its Light answer gives an
upper bound. Close threshold spacing makes these bounds differ by only a small multiplicative factor. Nodes that are Light in the first round
have small counts and can be estimated by zero.

The main issue is to account for all classification calls without introducing dependence on the number of rounds. We assign decreasing privacy budgets whose sums are bounded by $\varepsilon$ and $\delta$. 
Although smaller budgets increase the later classification margins, the thresholds grow faster, through our careful construction. This keeps each margin small relative to its threshold and makes error bound $O_{\alpha,\varepsilon,\delta}(\log\log h)$, an exponential improvement over the error of $O_{\alpha,\varepsilon,\delta}(\log h)$ in~\cite{ghazi2023}.

\subsection{Construction}

Fix $0<\varepsilon\le1$ and $0<\delta<1/10$. A call to \textsc{BetweenCut} below means invoking Algorithm~\ref{alg:tree} with its local
parameters calibrated to the specified privacy budget as in the main paper. 

\paragraph{A uniform classification bound.}
Because the privacy parameters will vary across calls, we first obtain a margin bound with one numerical constant valid throughout the parameter range. Let $E_h(\varepsilon',\delta')$ denote the calibrated deterministic classification margin. 
The explicit utility bound for Algorithm~\ref{alg:tree} gives, for every $h\ge3$, $0<\varepsilon'\le1$, and $0<\delta'<1/10$,
\begin{align}
    E_h(\varepsilon',\delta')
    \le 1+\frac{5}{\varepsilon'}\log\frac2{\delta'}
       \left(1+\log\frac{L}{\delta'}\right),
    \qquad L:=\left\lceil\log_2(h+1)\right\rceil.
\end{align}
To simplify this expression, write $\ell:=\log(1/\delta')>\log 10>2$ and $t:=\log\log h>0$.
Since the height $h\ge3$ is an integer,
$
L=1+\lfloor\log_2 h\rfloor\le 1+\frac{\log h}{\log 2}\le 3\log h.
$
Thus $\log L\le\log 3+t$, and the two logarithmic factors satisfy
\begin{align}
    \log\frac2{\delta'}
    &=\ell+\log 2
     \le\frac32\ell
    \\
    1+\log\frac{L}{\delta'}
    &\le\ell+t+1+\log 3\le 2\ell+t
    \le 2(\ell+t).
\end{align}
Substituting these bounds gives
\begin{align}
    E_h(\varepsilon',\delta')
    \le\frac{16}{\varepsilon'}\ell(\ell+t).
\end{align}
We therefore fix $C:=16$ once and define
\begin{equation}
    e_h(\varepsilon',\delta')
    :=\frac{C}{\varepsilon'}\log\frac1{\delta'}
       \left(\log\frac1{\delta'}+\log\log h\right).
    \label{eq:app-counting-margin}
\end{equation}
The derivation proves $E_h(\varepsilon',\delta')\le e_h(\varepsilon',\delta')$ for every allowed $h,\varepsilon',\delta'$---the same $C$ works for all calls. 

In particular, a call at threshold $\theta\ge0$ is
$(\varepsilon',\delta')$-DP and returns an ancestor-closed Heavy set
$\widehat H$ satisfying, simultaneously with probability one,
\begin{equation}
    \begin{aligned}
        u\in\widehat H
        &\ \Longrightarrow\
        c_u(D)\ge\theta-e_h(\varepsilon',\delta'),\\
        u\notin\widehat H
        &\ \Longrightarrow\
        c_u(D)<\theta+e_h(\varepsilon',\delta').
    \end{aligned}
    \label{eq:app-counting-classification}
\end{equation}
These are the lower and upper bounds on counts used by the reduction.

\paragraph{Thresholds and privacy budgets.}
We choose thresholds a factor $g>1$ apart and keep each classification margin at most a fraction $\gamma$ of its threshold. Set
\begin{align}
    \gamma:=\frac{\alpha}{8},\qquad
    g:=1+\frac{\alpha}{2},\qquad
    \rho:=g^{-1/2},\qquad p:=1-\rho.
\end{align}
Define the initial threshold by
\begin{align}
    \theta_1:=\frac{e_h(\varepsilon p,\delta p)}{\gamma}.
\end{align}
For every potential round $j\ge1$, use
\begin{align}
    \theta_j:=\theta_1g^{j-1},\qquad
    \varepsilon_j:=\varepsilon p\rho^{j-1},\qquad
    \delta_j:=\delta p\rho^{j-1}.
    \label{eq:app-counting-schedule}
\end{align}
It suffices to apply basic composition to demonstrate that the overall procedure achieves $(\varepsilon, \delta)$-DP. 
Since $0<\rho<1$ and $p=1-\rho$, the finite geometric-sum formula gives, for every integer $J\ge1$,
\begin{align}
    \sum_{j=1}^{J}\varepsilon_j
    &=\varepsilon(1-\rho^J)\le\varepsilon\\
    \sum_{j=1}^{J}\delta_j
    &=\delta(1-\rho^J)\le\delta.
    \label{eq:app-counting-budgets}
\end{align}
Thus, the budgets allocated to the executed calls never exceed $\varepsilon$ and $\delta$, even though their number is not fixed in advance. 

\paragraph{The counting procedure.}
For rounds $j=1,2,\ldots$, run \textsc{BetweenCut} at threshold $\theta_j$ with privacy budget $(\varepsilon_j,\delta_j)$.
Each call starts with fresh labels and randomness on the original tree and database. When a node $u$ is first classified Light, record its round $j(u)$ and permanently assign
\begin{align}
    \widetilde c_u:=
    \begin{cases}
        0, & j(u)=1,\\
        (1+\gamma)\theta_{j(u)}, & j(u)\ge2.
    \end{cases}
\end{align}
The procedure stops if the root is Light, and then releases the estimate vector. Denote this estimator by $\mathcal A_{\mathrm{count}}$.

\subsection{Guarantee and proof}

\begin{theorem}[Numerical tree counting from \textsc{BetweenCut}]
\label{thm:app-tree-counting}
For every finite public rooted tree of height $h\ge3$,
$0<\alpha<1$, $0<\varepsilon\le1$, and $0<\delta<1/10$, 
the estimator $\mathcal A_{\mathrm{count}}$ is $(\varepsilon,\delta)$-DP. Its nonnegative
estimates satisfy, simultaneously with probability one,
\begin{align}
    |\widetilde c_u-c_u(D)|
    \le\max\{\alpha c_u(D), B\}
    \qquad\text{for every }u\in V,
    \label{eq:app-counting-accuracy}
\end{align}
where
\begin{align}
    B:=(1+\gamma)\theta_1
      =O_{\alpha,\varepsilon,\delta}(\log\log h).
\end{align}
Consequently,
\begin{align}
    \operatorname{mRMSE}_{\alpha}(\mathcal A_{\mathrm{count}};T)
    \le B
    =O_{\alpha,\varepsilon,\delta}(\log\log h).
    \label{eq:app-counting-accuracy-stated}
\end{align}
\end{theorem}

\vspace{0pt}
\noindent\textbf{Comparison with Ghazi et al.~\cite{ghazi2023}.} Compared with $\operatorname{mRMSE}_{\alpha}(\mathcal A_{\mathrm G};T)=O_{\alpha,\varepsilon,\delta}(\log h)$ from~\cite{ghazi2023}, 
Eq.~\eqref{eq:app-counting-accuracy-stated} is an exponential improvement in $h$.
Ghazi et al.~\cite{ghazi2023} also considers a threshold-based utility notion with a public accuracy threshold $\tau_u$ for each node $u$, requiring $
    |\widetilde c_u-c_u(D)|
    \le \alpha\max\{c_u(D),\tau_u\}.$
Our Theorem~\ref{thm:app-tree-counting} directly gives
$|\widetilde c_u-c_u(D)|\le\alpha \max\{c_u(D),\frac{B}{\alpha}\}$, 
without requiring such node-specific thresholds.

% The simultaneous guarantee in Eq.~\eqref{eq:app-counting-accuracy} is actually stronger
% than the stated root-mean-square bound in Eq.~\eqref{eq:app-counting-accuracy-stated}. It gives absolute error at most
% $B$ for counts at most $B/\alpha$, and relative error at most $\alpha$
% for larger counts.

\begin{proof}[Proof of Theorem~\ref{thm:app-tree-counting}]
\emph{Classification accuracy at every round.}
The function $e_h$ defined in Eq.~\eqref{eq:app-counting-margin} upper-bounds the calibrated deterministic margin.
It suffices to show that
$e_h(\varepsilon_j,\delta_j)\le\gamma\theta_j$.

Write
\begin{align}
    A:=\log\frac1{\delta p},\qquad t:=\log\log h.
\end{align}
Since $0<p<1$ and $\delta<1/10$, we have $A>\log 10>2$. Also, $t>0$ because $h\ge3$. Fix a round $j\ge1$ and set
\begin{align}
    x:=(j-1)\log\frac1\rho\ge0.
\end{align}
The construction schedule in Eq.~\eqref{eq:app-counting-schedule} and $\rho=g^{-1/2}$ give
\begin{align}
    \varepsilon_j=\varepsilon p e^{-x},\qquad
    \log\frac1{\delta_j}=A+x,\qquad
    \theta_j=\theta_1e^{2x}.
\end{align}
Using the definition of $e_h$, we obtain
\begin{align}
    e_h(\varepsilon_j,\delta_j)
    =\frac{C}{\varepsilon p e^{-x}}(A+x)(A+x+t)
    =\frac{C}{\varepsilon p}e^x(A+x)(A+x+t).
\end{align}
The choice
$\theta_1=e_h(\varepsilon p,\delta p)/\gamma$ yields
\begin{align}
    \gamma\theta_j
    =\gamma\theta_1e^{2x}
    =e_h(\varepsilon p,\delta p)e^{2x}
    =\frac{C}{\varepsilon p}e^{2x}A(A+t).
\end{align}
Since $t>1$, $1+y < e^y$ for any $y$, and $A>2$, we have
\begin{align}
    \frac{e_h(\varepsilon_j,\delta_j)}{\gamma\theta_j}
    &=\frac{
        \frac{C}{\varepsilon p}e^x(A+x)(A+x+t)
    }{
        \frac{C}{\varepsilon p}e^{2x}A(A+t)
    }\\
    &=e^{-x}\frac{A+x}{A}\frac{A+x+t}{A+t}\\
    &=e^{-x}\left(1+\frac{x}{A}\right)
             \left(1+\frac{x}{A+t}\right)\\
    &\le e^{-x}\left(1+\frac{x}{A}\right)^2\\
    &\le e^{-x}e^{2x/A}\\
    &\le1.
\label{eq:app-counting-margin-ratio}
\end{align}
Consequently, every call's deterministic classification margin is
at most $e_h(\varepsilon_j,\delta_j)\le\gamma\theta_j$, with probability $1$.

% For any fixed database, each potential call satisfies
% \eqref{eq:app-counting-classification} with probability one.
% Since there are countably many potential calls, their guarantees
% hold on a common probability-one event. We use this event for
% termination and numerical accuracy below; the privacy argument
% concerns unconditional output distributions.

\noindent\emph{Termination.}
Fix a finite database $D$ and write $n:=|D|$. By our construction, $(1-\gamma)\theta_j$ is increasing in each round by a factor of $(1+\frac{\alpha}{2})$. 
Hence, 
\begin{align}
    J_D:=\min\{j\ge1:(1-\gamma)\theta_j>n\}
\end{align}
is finite. On the probability-one event above, if round $J_D$ is
reached, then
\begin{align*}
    n<(1-\gamma)\theta_{J_D}
    \le\theta_{J_D}
       -e_h(\varepsilon_{J_D},\delta_{J_D}),
\end{align*}
according to what we just proved---$e_h(\varepsilon_j,\delta_j)\le\gamma\theta_j$. Namely, the root will be classified as Light at round $J_D$. The procedure therefore stops after at most $J_D$ calls.

\noindent\emph{Simultaneous numerical accuracy.}
Fix a database $D$ and work on the probability-one event where all classification guarantees hold. For each node $u$, let $j(u)$ be its first Light round, which exists by the above termination argument.

If $j(u)=1$, then
\begin{align}
    |\widetilde c_u-c_u(D)|
    &=c_u(D)
    <\theta_1+e_h(\varepsilon_1,\delta_1)
    \le(1+\gamma)\theta_1
    =B
\end{align}
Otherwise, write $j:=j(u)\ge2$ for readability. The node was Heavy in round $j-1$
and is Light in round $j$. Therefore,
\begin{align}
    c_u(D)
    &\ge\theta_{j-1}-e_h(\varepsilon_{j-1},\delta_{j-1})
    &&\text{(Heavy in round $j-1$)}\\
    &\ge(1-\gamma)\theta_{j-1}
    &&\text{(margin bound in round $j-1$)},\\
    c_u(D)
    &<\theta_j+e_h(\varepsilon_j,\delta_j)
    &&\text{(Light in round $j$)}\\
    &\le(1+\gamma)\theta_j=:\widetilde c_u
    &&\text{(margin bound in round $j$)}.
\end{align}
The lower bound is positive because $\gamma<1$ and
$\theta_{j-1}>0$. We obtain
\begin{align}
    \frac{\widetilde c_u}{c_u(D)}
    &\le\frac{(1+\gamma)\theta_j}
                  {(1-\gamma)\theta_{j-1}}=\frac{g(1+\gamma)}{1-\gamma}
    =\frac{(1+\alpha/2)(1+\alpha/8)}{1-\alpha/8}
    =1+\alpha-
      \frac{\alpha(4-3\alpha)}{16(1-\alpha/8)}
    \le1+\alpha.
\end{align}
Together with $\widetilde c_u\ge c_u(D)$, this implies
\begin{align}
    |\widetilde c_u-c_u(D)|
    =\widetilde c_u-c_u(D)
    \le\alpha c_u(D).
\end{align}
Combining the two cases proves
Eq.~\eqref{eq:app-counting-accuracy} simultaneously for all nodes.

Expanding the definition of $B$ gives
\begin{align}
    B=(1+\gamma)\theta_1
     =\frac{(1+\gamma)C}{\gamma\varepsilon p}
       A\bigl(A+\log\log h\bigr).
\end{align}
For fixed $\alpha,\varepsilon,\delta$, the quantities $\gamma$, $p$,
and $A$ are fixed and independent of both $h$ and $D$. Consequently,
$B=O_{\alpha,\varepsilon,\delta}(\log\log h)$, with no database-size
dependence.

Therefore, for every database and node
\begin{align*}
    r_u(D)
    &=\bigl(|\widetilde c_u-c_u(D)|-\alpha c_u(D)\bigr)_+\\
    &\le\bigl(\max\{B,\alpha c_u(D)\}-\alpha c_u(D)\bigr)_+
    &&\text{by Eq.~\eqref{eq:app-counting-accuracy}}\\
    &=\bigl(B-\alpha c_u(D)\bigr)_+\\
    &\le B.
\end{align*}
Thus 
\begin{align}
    \operatorname{mRMSE}_{\alpha}(\mathcal A_{\mathrm{count}};T)
    =\sup_D\max_{u\in V}\sqrt{\mathbb E[r_u(D)^2]}
    \le B
    =O_{\alpha,\varepsilon,\delta}(\log\log h).
\end{align}
\end{proof}

\paragraph{Fixed tree height.}
The assumption $h\ge3$ concerns the height of the original public tree $T$, not the number of levels still awaiting classification.
Every counting round invokes Algorithm~\ref{alg:tree} on the same tree with the same height parameter $h$. Thus the assumption $h\ge 3$ remains valid throughout. Each call starts with fresh Heavy/Light \emph{labels}, while previously assigned \emph{numerical estimates} remain unchanged.

\paragraph{Increasing thresholds.}
The classifier in~\cite{ghazi2023} depends on an upper bound $M$ on the root count through the ratio $M/\tau$. Their decreasing-threshold reduction controls this ratio by continuing on lighter subtrees, whose root-count bounds decrease along with the thresholds~\cite{ghazi2023}.
In contrast, \textsc{BetweenCut}'s classification margin does not depend on a root-count bound, so we can repeatedly classify the original tree without this subtree refinement. Increasing thresholds let us start at a public, database-size-independent accuracy scale and stop when the root becomes Light, without first estimating an upper bound on the database size. This direction simplifies the construction.

%% file: references.bib
@InProceedings{tcc06,
author="Dwork, Cynthia
and McSherry, Frank
and Nissim, Kobbi
and Smith, Adam",
title="Calibrating Noise to Sensitivity in Private Data Analysis",
booktitle="TCC",
year="2006"
}

@InProceedings{eps-delta-indistinguish,
author="Dwork, Cynthia
and Kenthapadi, Krishnaram
and McSherry, Frank
and Mironov, Ilya
and Naor, Moni",
title="Our Data, Ourselves: Privacy Via Distributed Noise Generation",
booktitle="EUROCRYPT",
year="2006",
}

@inproceedings{ghazi2023,
  author        = {Ghazi, Badih and Kamath, Pritish and Kumar, Ravi and
                   Manurangsi, Pasin and Wu, Kewen},
  title         = {On Differentially Private Counting on Trees},
  booktitle     = {ICALP},
  year          = {2023}
}

@inproceedings{lumberjack2026,
  author        = {Lebeda, Christian Janos and Erb, David and Cebere, Tudor
                   and Bellet, Aur{\'e}lien},
  title         = {{Lumberjack}: Better Differentially Private Random Forests
                   through Heavy Hitter Detection in Trees},
  booktitle     = {NeurIPS},
  year          = {2026}
}

@inproceedings{cohenlyu2023,
  author        = {Cohen, Edith and Lyu, Xin},
  title         = {The Target-Charging Technique for Privacy Analysis
                   across Interactive Computations},
  booktitle     = {NeurIPS},
  year          = {2023}
}

@inproceedings{bnsv2015,
  author        = {Bun, Mark and Nissim, Kobbi and Stemmer, Uri
                   and Vadhan, Salil},
  title         = {Differentially Private Release and Learning
                   of Threshold Functions},
  booktitle     = {FOCS},
  pages         = {634--649},
  year          = {2015}
}

@article{biswas2026,
author = {Biswas, Ari and Cormode, Graham and Kanza, Yaron and Srivastava, Divesh and Zhou, Zhengyi},
title = {Differentially Private Hierarchical Heavy Hitters},
year = {2024},
volume = {2},
number = {5},
journal = {PACMMOD},
month = nov,
numpages = {25},
}

@inproceedings{wu2022heavyhitters,
  author    = {Wu, Hao and Wirth, Anthony},
  title     = {Asymptotically Optimal Locally Private Heavy Hitters
               via Parameterized Sketches},
  booktitle = {AISTATS},
  year      = {2022}
}

@inproceedings{zhu2020federated,
  author    = {Zhu, Wennan and Kairouz, Peter and McMahan, Brendan
               and Sun, Haicheng and Li, Wei},
  title     = {Federated Heavy Hitters Discovery with Differential Privacy},
  booktitle = {AISTATS},
  year      = {2020}
}

@inproceedings{kim2021ngram,
  author    = {Kim, Kunho and Gopi, Sivakanth and Kulkarni, Janardhan
               and Yekhanin, Sergey},
  title     = {Differentially Private {$n$}-gram Extraction},
  booktitle = {NIPS},
  pages     = {5102--5111},
  year      = {2021},
}

@inproceedings{zhang2016privtree,
  author    = {Zhang, Jun and Xiao, Xiaokui and Xie, Xing},
  title     = {{PrivTree}: A Differentially Private Algorithm for
               Hierarchical Decompositions},
  booktitle = {SIGMOD},
  pages     = {155--170},
  year      = {2016},
}

@article{dwork:roth:14,
  author       = {Cynthia Dwork and
                  Aaron Roth},
  title        = {The Algorithmic Foundations of Differential Privacy},
  journal      = {FTCS},
  year         = {2014}
}

@INPROCEEDINGS{boosting-and-dp,
  author={Dwork, Cynthia and Rothblum, Guy N. and Vadhan, Salil},
  booktitle={FOCS}, 
  title={Boosting and Differential Privacy}, 
  year={2010}}

@InProceedings{pmlr-v37-kairouz15,
  title = 	 {The Composition Theorem for Differential Privacy},
  author = 	 {Kairouz, Peter and Oh, Sewoong and Viswanath, Pramod},
  booktitle = 	 {ICML},
  year = 	 {2015}
}

@inproceedings{optimal-dp-thresholds,
author = {Cohen, Edith and Lyu, Xin and Nelson, Jelani and Sarl{\'o}s, Tam{\'a}s and Stemmer, Uri},
title = {Optimal Differentially Private Learning of Thresholds and Quasi-Concave Optimization},
year = {2023},
booktitle = {STOC},
pages = {472–482},
}

@inproceedings{svt-origin,
author = {Dwork, Cynthia and Naor, Moni and Reingold, Omer and Rothblum, Guy N. and Vadhan, Salil},
title = {On the complexity of differentially private data release: efficient algorithms and hardness results},
year = {2009},
booktitle = {STOC}}

@InProceedings{svt-app-pmlr-v117-nandi20a,
  title = 	 {Privately Answering Classification Queries in the Agnostic PAC Model},
  author =       {Nandi, Anupama and Bassily, Raef},
  booktitle = 	 {ALT},
  year = 	 {2020},
}

@inproceedings{svt-app-mw-linear,
author = {Ullman, Jonathan},
title = {Private Multiplicative Weights Beyond Linear Queries},
year = {2015},
booktitle = {PODS}
}

@InProceedings{svt-app-pmlr-v65-feldman17a,
  title = 	 {Generalization for Adaptively-chosen Estimators via Stable Median},
  author = 	 {Feldman, Vitaly and Steinke, Thomas},
  booktitle = 	 {COLT},
  year = 	 {2017}
}

@inproceedings{svt-app-valid,
author = {Dwork, Cynthia and Feldman, Vitaly and Hardt, Moritz and Pitassi, Toniann and Reingold, Omer and Roth, Aaron Leon},
title = {Preserving Statistical Validity in Adaptive Data Analysis},
year = {2015},
booktitle = {STOC}
}

@InProceedings{svt-app-pmlr-v83-nissim18a,
  title = 	 {Clustering Algorithms for the Centralized and Local Models},
  author = 	 {Nissim, Kobbi and Stemmer, Uri},
  booktitle = 	 {ALT},
  year = 	 {2018}
}

@inproceedings{svt-app-dp-level-erm,
author = {Ligett, Katrina and Neel, Seth and Roth, Aaron and Waggoner, Bo and Wu, Zhiwei Steven},
title = {Accuracy first: selecting a differential privacy level for accuracy-constrained ERM},
year = {2017},
booktitle = {NIPS}
}

@inproceedings{svt-app-10.1145/2488608.2488651,
author = {Hsu, Justin and Roth, Aaron and Ullman, Jonathan},
title = {Differential privacy for the analyst via private equilibrium computation},
year = {2013},
booktitle = {STOC}
}

@inproceedings{svtrevisited,
  author       = {Haim Kaplan and
                  Yishay Mansour and
                  Uri Stemmer},
  title        = {The Sparse Vector Technique, Revisited},
  booktitle = {COLT},
  year         = {2021}
}

@article{understandingsvt,
  author       = {Min Lyu and
                  Dong Su and
                  Ninghui Li},
  title        = {Understanding the Sparse Vector Technique for Differential Privacy},
  journal      = {PVLDB},
  year         = {2017}
}

@inproceedings{korolova2009,
  author    = {Aleksandra Korolova and Krishnaram Kenthapadi and
               Nina Mishra and Alexandros Ntoulas},
  title     = {Releasing Search Queries and Clicks Privately},
  booktitle = {WWW},
  year      = {2009}
}

@inproceedings{sandtdp,
  author = {Bharadwaj, Akash and Cormode, Graham},
  title = {Sample and Threshold Differential Privacy: Histograms and Applications},
  booktitle = {AISTATS},
  year = {2022}
}

@inproceedings{DBLP:conf/icml/KaplanSS22,
  author       = {Haim Kaplan and
                  Shachar Schnapp and
                  Uri Stemmer},
  title        = {Differentially Private Approximate Quantiles},
  booktitle    = {ICML},
  year         = {2022}
}

@article{dp-sparse-linear,
author = {Kerschbaum, Florian and Lee, Steven and Wu, Hao},
title = {Optimal Pure Differentially Private Sparse Histograms in Deterministic Linear Time},
year = {2026},
volume = {4},
number = {2},
journal = {PACMMOD},
month = may
}

@inproceedings{DBLP:conf/colt/KaplanLMNS20,
  author       = {Haim Kaplan and
                  Katrina Ligett and
                  Yishay Mansour and
                  Moni Naor and
                  Uri Stemmer},
  title        = {Privately Learning Thresholds: Closing the Exponential Gap},
  booktitle    = {COLT},
  year         = {2020}
}

@InProceedings{lebeda_et_al:LIPIcs.FORC.2025.23,
  author    = {Lebeda, Christian Janos and Retschmeier, Lukas},
  title     = {{The Correlated Gaussian Sparse Histogram Mechanism}},
  booktitle = {FORC},
  year      = {2025}
}

@inproceedings{dp-sparse-data,
author = {Cormode, Graham and Procopiuc, Cecilia and Srivastava, Divesh and Tran, Thanh T. L.},
title = {Differentially private summaries for sparse data},
year = {2012},
booktitle = {ICDT},
pages = {299–311}
}

@INPROCEEDINGS{dp-spatial,
  author={Cormode, Graham and Procopiuc, Cecilia and Srivastava, Divesh and Shen, Entong and Yu, Ting},
  booktitle={ICDE}, 
  title={Differentially Private Spatial Decompositions}, 
  year={2012},
  pages={20-31}
  }

@inproceedings{instance-mean,
author = {Huang, Ziyue and Liang, Yuting and Yi, Ke},
title = {Instance-optimal mean estimation under differential privacy},
year = {2021},
booktitle = {NeurIPS}
}

@inproceedings{
blocki2023differentially,
title={Differentially Private $L_2$-Heavy Hitters in the Sliding Window Model},
author={Jeremiah Blocki and Seunghoon Lee and Tamalika Mukherjee and Samson Zhou},
booktitle={ICLR},
year={2023}
}

@ARTICLE{ldp-hh-wang,
  author={Wang, Tianhao and Li, Ninghui and Jha, Somesh},
  journal={TDSC}, 
  title={Locally Differentially Private Heavy Hitter Identification}, 
  year={2021},
  volume={18},
  number={2},
  pages={982-993}
  }

@inproceedings{differentially-private-heavy-fa,
title = {Differentially Private Heavy Hitter Detection using Federated Analytics},
booktitle = {ICML},
author = {Karan Chadha and Junye Chen and John Duchi and Vitaly Feldman and Hanieh Hashemi and Omid Javidbakht and Audra McMillan and and Kunal Talwar},
year = {2023}
}

@article{federated-hh-ldp,
author = {Zhang, Yuemin and Ye, Qingqing and Hu, Haibo},
title = {Federated Heavy Hitter Analytics with Local Differential Privacy},
year = {2025},
volume = {3},
number = {1},
journal = {PACMMOD},
}

@article{hierarchical-spatial-ldp,
author = {Alptekin, Ece and Balioglu, Berkay Kemal and Gursoy, M. Emre},
title = {Hierarchical Spatial Decompositions under Local Differential Privacy},
year = {2025},
issue_date = {August 2025},
volume = {19},
number = {7},
journal = {TKDD},
month = aug,
articleno = {128},
numpages = {37}
}

@InProceedings{rectangle-queries,
author="Dwork, Cynthia
and Naor, Moni
and Reingold, Omer
and Rothblum, Guy N.",
editor="Iwata, Tetsu
and Cheon, Jung Hee",
title="Pure Differential Privacy for Rectangle Queries via Private Partitions",
booktitle="ASIACRYPT",
year="2015",
pages="735--751"
}

@article{Xiao2012DPCubeDP,
  title={DPCube: Differentially Private Histogram Release through Multidimensional Partitioning},
  author={Yonghui Xiao and Li Xiong and Liyue Fan and Slawomir Goryczka and Haoran Li},
  journal={TDP},
  year={2014}
}

@inproceedings{
imola2025differentially,
title={Differentially Private Quantiles with Smaller Error},
author={Jacob Imola and Fabrizio Boninsegna and Hannah Keller and Anders Aamand and Amrita Roy Chowdhury and Rasmus Pagh},
booktitle={NeurIPS},
year={2025}
}

@article{
alabi2023bounded,
title={Bounded Space Differentially Private Quantiles},
author={Daniel Alabi and Omri Ben-Eliezer and Anamay Chaturvedi},
journal={TMLR},
year={2023}
}

@InProceedings{pmlr-v139-gillenwater21a,
  title = 	 {Differentially Private Quantiles},
  author =       {Gillenwater, Jennifer and Joseph, Matthew and Kulesza, Alex},
  booktitle = 	 {ICML},
  year = 	 {2021},
}

@inproceedings{DP_unbounded,
 author = {Durfee, David},
 booktitle = {NeurIPS},
 title = {Unbounded Differentially Private Quantile and Maximum Estimation},
 year = {2023}
}

@inproceedings{practical-ldp-hh,
author = {Bassily, Raef and Nissim, Kobbi and Stemmer, Uri and Thakurta, Abhradeep},
title = {Practical locally private heavy hitters},
year = {2017},
booktitle = {NIPS},
pages = {2285–2293}
}

@inproceedings{chang2021oneround,
  author    = {Alisa Chang and Badih Ghazi and Ravi Kumar and
               Pasin Manurangsi},
  title     = {Locally Private {$k$}-Means in One Round},
  booktitle = {ICML},
  year      = {2021},
}

@misc{apple2017learning,
  author       = {{Differential Privacy Team}},
  title        = {Learning with Privacy at Scale},
  year         = {2017},
  month        = dec,
  day          = {6},
  organization = {Apple Machine Learning Research},
  url          = {https://machinelearning.apple.com/research/learning-with-privacy-at-scale},
  note         = {Accessed: 2026-10-05}
}

@inproceedings{ferry-random-forest-attack,
author = {Ferry, Julien and Fukasawa, Ricardo and Pascal, Timoth{\'e}e and Vidal, Thibaut},
title = {Trained random forests completely reveal your dataset},
year = {2024},
booktitle = {ICML},
numpages = {25}
}
